\documentclass{llncs}

\usepackage[utf8]{inputenc}
\usepackage{amsmath,amssymb,latexsym}
\usepackage{url}
\usepackage{subfigure}
\usepackage{multirow}

\usepackage{tikz}
\usetikzlibrary{automata}
\usetikzlibrary{shapes.geometric}
\usetikzlibrary{intersections}
\tikzset{run/.style={minimum size=0.7cm,inner sep=2pt}}

\newcommand{\X}{\mathsf{X}}     %next
\newcommand{\U}{{\,\uU\,}}      %until
\newcommand{\uU}{\mathsf{U}}    %until without spaces
\newcommand{\F}{\mathsf{F}}     %in future
\newcommand{\G}{\mathsf{G}}     %globally
\newcommand{\Fs}{{\F\!\s}}      %striktni in future
\newcommand{\Gs}{{\G\s}}        %striktni globally
\newcommand{\s}{_\mathsf{s}}    %striktnitko

\newcommand{\true}{\textrm{{\it tt}}}

\newcommand{\AP}{\mathit{A\hskip-0.1ex P}}
\newcommand{\LTL}{{\mathrm{LTL}}}
\newcommand{\suf}[1]{_{#1..}}

\newcommand{\mA}{\mathcal{A}}
\newcommand{\mB}{\mathcal{B}}
\newcommand{\mD}{\mathcal{D}}
\newcommand{\mG}{\mathcal{G}}
\newcommand{\mGR}{\mathcal{GR}}

\newcommand{\mT}{\mathcal{T}}

\newcommand{\mO}{\mathcal{O}}
\newcommand{\mZ}{\mathcal{Z}}
\newcommand{\A}{_{\mA}}			%A-index
\newcommand{\T}{_{\mT}}			%T-index

\newcommand{\Inf}{\mathit{Inf\!}}
\newcommand{\AC}{\mathrm{AC}}
\newcommand{\AT}[1]{\mathrm{AT}_{\!#1}}
\newcommand{\must}{\mathit{must}}

\newcommand{\targets}{\mathit{targets}}

\DeclareMathOperator*{\range}{{range}}
\DeclareMathOperator*{\dom}{{dom}}

\usepackage{xcolor}

\makeatletter

\pgfkeys{/pgf/.cd,
  rectangle corner radius/.initial=3pt
}
\newif\ifpgf@rectanglewrc@donecorner@
\def\pgf@rectanglewithroundedcorners@docorner#1#2#3#4{%
  \edef\pgf@marshal{%
    \noexpand\pgfintersectionofpaths
      {%
        \noexpand\pgfpathmoveto{\noexpand\pgfpoint{\the\pgf@xa}{\the\pgf@ya}}%
        \noexpand\pgfpathlineto{\noexpand\pgfpoint{\the\pgf@x}{\the\pgf@y}}%
      }%
      {%
        \noexpand\pgfpathmoveto{\noexpand\pgfpointadd
          {\noexpand\pgfpoint{\the\pgf@xc}{\the\pgf@yc}}%
          {\noexpand\pgfpoint{#1}{#2}}}%
        \noexpand\pgfpatharc{#3}{#4}{\cornerradius}%
      }%
    }%
  \pgf@process{\pgf@marshal\pgfpointintersectionsolution{1}}%
  \pgf@process{\pgftransforminvert\pgfpointtransformed{}}%
  \pgf@rectanglewrc@donecorner@true
}
\pgfdeclareshape{rectangle with rounded corners}
{
  \inheritsavedanchors[from=rectangle] % this is nearly a rectangle
  \inheritanchor[from=rectangle]{north}
  \inheritanchor[from=rectangle]{north west}
  \inheritanchor[from=rectangle]{north east}
  \inheritanchor[from=rectangle]{center}
  \inheritanchor[from=rectangle]{west}
  \inheritanchor[from=rectangle]{east}
  \inheritanchor[from=rectangle]{mid}
  \inheritanchor[from=rectangle]{mid west}
  \inheritanchor[from=rectangle]{mid east}
  \inheritanchor[from=rectangle]{base}
  \inheritanchor[from=rectangle]{base west}
  \inheritanchor[from=rectangle]{base east}
  \inheritanchor[from=rectangle]{south}
  \inheritanchor[from=rectangle]{south west}
  \inheritanchor[from=rectangle]{south east}

  \savedmacro\cornerradius{%
    \edef\cornerradius{\pgfkeysvalueof{/pgf/rectangle corner radius}}%
  }

  \backgroundpath{%
    \northeast\advance\pgf@y-\cornerradius\relax
    \pgfpathmoveto{}%
    \pgfpatharc{0}{90}{\cornerradius}%
    \northeast\pgf@ya=\pgf@y\southwest\advance\pgf@x\cornerradius\relax\pgf@y=\pgf@ya
    \pgfpathlineto{}%
    \pgfpatharc{90}{180}{\cornerradius}%
    \southwest\advance\pgf@y\cornerradius\relax
    \pgfpathlineto{}%
    \pgfpatharc{180}{270}{\cornerradius}%
    \northeast\pgf@xa=\pgf@x\advance\pgf@xa-\cornerradius\southwest\pgf@x=\pgf@xa
    \pgfpathlineto{}%
    \pgfpatharc{270}{360}{\cornerradius}%
    \northeast\advance\pgf@y-\cornerradius\relax
    \pgfpathlineto{}%
  }

  \anchor{before north east}{\northeast\advance\pgf@y-\cornerradius}
  \anchor{after north east}{\northeast\advance\pgf@x-\cornerradius}
  \anchor{before north west}{\southwest\pgf@xa=\pgf@x\advance\pgf@xa\cornerradius
    \northeast\pgf@x=\pgf@xa}
  \anchor{after north west}{\northeast\pgf@ya=\pgf@y\advance\pgf@ya-\cornerradius
    \southwest\pgf@y=\pgf@ya}
  \anchor{before south west}{\southwest\advance\pgf@y\cornerradius}
  \anchor{after south west}{\southwest\advance\pgf@x\cornerradius}
  \anchor{before south east}{\northeast\pgf@xa=\pgf@x\advance\pgf@xa-\cornerradius
    \southwest\pgf@x=\pgf@xa}
  \anchor{after south east}{\southwest\pgf@ya=\pgf@y\advance\pgf@ya\cornerradius
    \northeast\pgf@y=\pgf@ya}

  \anchorborder{%
    \pgf@xb=\pgf@x% xb/yb is target
    \pgf@yb=\pgf@y%
    \southwest%
    \pgf@xa=\pgf@x% xa/ya is se
    \pgf@ya=\pgf@y%
    \northeast%
    \advance\pgf@x by-\pgf@xa%
    \advance\pgf@y by-\pgf@ya%
    \pgf@xc=.5\pgf@x% x/y is half width/height
    \pgf@yc=.5\pgf@y%
    \advance\pgf@xa by\pgf@xc% xa/ya becomes center
    \advance\pgf@ya by\pgf@yc%
    \edef\pgf@marshal{%
      \noexpand\pgfpointborderrectangle
      {\noexpand\pgfqpoint{\the\pgf@xb}{\the\pgf@yb}}
      {\noexpand\pgfqpoint{\the\pgf@xc}{\the\pgf@yc}}%
    }%
    \pgf@process{\pgf@marshal}%
    \advance\pgf@x by\pgf@xa% 
    \advance\pgf@y by\pgf@ya%
    \pgfextract@process\borderpoint{}%
    \pgf@rectanglewrc@donecorner@false
    \southwest\pgf@xc=\pgf@x\pgf@yc=\pgf@y
    \advance\pgf@xc\cornerradius\relax\advance\pgf@yc\cornerradius\relax 
    \borderpoint
    \ifdim\pgf@x<\pgf@xc\relax\ifdim\pgf@y<\pgf@yc\relax
      \pgf@rectanglewithroundedcorners@docorner{-\cornerradius}{0pt}{180}{270}%
    \fi\fi
    \ifpgf@rectanglewrc@donecorner@\else
      \southwest\pgf@yc=\pgf@y\relax\northeast\pgf@xc=\pgf@x\relax
      \advance\pgf@xc-\cornerradius\relax\advance\pgf@yc\cornerradius\relax
      \borderpoint
      \ifdim\pgf@x>\pgf@xc\relax\ifdim\pgf@y<\pgf@yc\relax
       \pgf@rectanglewithroundedcorners@docorner{0pt}{-\cornerradius}{270}{360}%
      \fi\fi
    \fi
    \ifpgf@rectanglewrc@donecorner@\else
      \northeast\pgf@xc=\pgf@x\relax\pgf@yc=\pgf@y\relax
      \advance\pgf@xc-\cornerradius\relax\advance\pgf@yc-\cornerradius\relax
      \borderpoint
      \ifdim\pgf@x>\pgf@xc\relax\ifdim\pgf@y>\pgf@yc\relax
       \pgf@rectanglewithroundedcorners@docorner{\cornerradius}{0pt}{0}{90}%
      \fi\fi
    \fi
    \ifpgf@rectanglewrc@donecorner@\else
      \northeast\pgf@yc=\pgf@y\relax\southwest\pgf@xc=\pgf@x\relax
      \advance\pgf@xc\cornerradius\relax\advance\pgf@yc-\cornerradius\relax
      \borderpoint
      \ifdim\pgf@x<\pgf@xc\relax\ifdim\pgf@y>\pgf@yc\relax
       \pgf@rectanglewithroundedcorners@docorner{0pt}{\cornerradius}{90}{180}%
      \fi\fi
    \fi
  }
}

\makeatother

\begin{document}
\frontmatter

%{{{ Title + Author

% \title{Effective Translation of LTL to\\ Deterministic Rabin
% Automata:\\ Beyond the (F,G)-Fragment}
% \title{May/Must Alternating Automata and Translation of LTL($\Fs,\Gs$) to
%   Deterministic (Transition-Based Generalized) Rabin Automata}
\title{Effective Translation of LTL to Deterministic Rabin Automata: Beyond
  the (F,G)-Fragment\thanks{This is a full version of the paper accepted to
    ATVA 2013.}}
  % \thanks{The authors are supported by The Czech
  % Science Foundation, grant P202/12/G061.}}
\author{Tom\'{a}\v{s} Babiak \and Franti\v{s}ek Blahoudek \and Mojm\'{i}r
  K\v{r}et\'{i}nsk\'{y} \and Jan Strej\v{c}ek}
\institute{Faculty of Informatics, Masaryk University, Brno, Czech Republic\\
%  Botanick\'{a} 68a, 60200 Brno, Czech Republic \\
  \email{\{xbabiak,\,xblahoud,\,kretinsky,\,strejcek\}@fi.muni.cz}}

%}}}

\maketitle
%\sloppy

%{{{ Abstract

\begin{abstract}
  Some applications of linear temporal logic (LTL) require to translate
  formulae of the logic to deterministic $\omega$-automata. There are
  currently two translators producing deterministic automata:
  \texttt{ltl2dstar} working for the whole LTL and Rabinizer applicable to
  LTL($\F,\G$) which is the LTL fragment using only modalities $\F$ and $\G$.  We
  present a new translation to deterministic Rabin automata via alternating
  automata and deterministic transition-based generalized Rabin
  automata. Our translation applies to a fragment that is strictly larger
  than LTL($\F,\G$). Experimental results show that our algorithm can produce
  significantly smaller automata compared to Rabinizer and
  \texttt{ltl2dstar}, especially for more complex LTL formulae.
  % Zavedeme novou tridu May/Must alternujicich automatu (MMAA) a ukazeme, ze
  % primo koresponduje fragmentu LTL omezeneho na striktni verze operatoru
  % eventually ($\Fs$) a globally ($\Gs$). Dale ukazeme, ze MMAA lze snadno
  % prelozit na deterministicke Rabinovi automaty. Preklad zmineneho fragmentu
  % LTL na MMAA a dale MMAA na Rabinovy automaty byl implementovan a porovnan
  % s nastroji ltl2dstar a Rabinizer. Experimentalni vysledky ukazuji, ze
  % vznikly prekladac podobne efektivni jako nastroj Rabinizer, ale preklada
  % vetsi fragment LTL.
\end{abstract}	

%}}}
%{{{ Intro

\section{Introduction}

\emph{Linear temporal logic (LTL)} is a popular formalism for specification
of behavioral system properties with major applications in the area of model
checking~\cite{CGP99,BK08}. More precisely, LTL is typically used as a
human-oriented front-end formalism as LTL formulae are succinct and easy to
write and understand. Model checking algorithms usually work with an
$\omega$-automaton representing all behaviors violating a given
specification formula rather than with the LTL formula directly. Hence,
specifications written in the form of LTL formulae are negated and
translated to equivalent $\omega$-automata~\cite{VardiW86}. There has been a
lot of attention devoted to translation of LTL to \emph{nondeterministic
  B\"uchi automata (NBA)}, see for example~\cite{Cou99,DGV99,SB00,GO01} and
the research in this direction still continues~\cite{DL11,BKRS12,BBDL13}.  However,
there are algorithms that need specifications given by \emph{deterministic}
$\omega$-automata, for example, those for LTL model checking of
probabilistic systems~\cite{Vardi85,CY95,BK08} and those for synthesis of
reactive modules for LTL specifications~\cite{Church62,PnueliR89}, for a
recent survey see \cite{Kupferman12}.
% \kickoff{Sem dopsat i nejaky algoritmus na syntezu i s odkazem.} 
As \emph{deterministic B\"uchi automata (DBA)} cannot express all the
properties expressible in LTL, one has to choose deterministic automata with
different acceptance condition.
% , for example Rabin, Muller, Streett, or parity
% acceptance. In this paper 
% Here we follow the most common choice: we study
% translation of LTL to \emph{deterministic Rabin automata (DRA)}.

There are basically two approaches to translation of LTL to
deterministic $\omega$-automata. The first one translates LTL to NBA
%using some standard algorithm 
and then it employs Safra's construction~\cite{Saf88} (or some of its
variants or alternatives like~\cite{Pit07,Sch09}) to transform the NBA into
a deterministic automaton. This approach is represented by the
% tool \texttt{ltl2dstar}~\footnote{\url{http://www.ltl2dstar.de}} 
tool \texttt{ltl2dstar}~\cite{Kle} which uses an improved Safra's
construction~\cite{KB06,KB07} usually in connection with LTL to NBA
translator LTL2BA~\cite{GO01}. The main advantage of this approach is its
universality: as LTL2BA can translate any LTL formula into an NBA and the
Safra's construction can transform any NBA to a \emph{deterministic Rabin
  automaton (DRA)}, \texttt{ltl2dstar} works for the whole LTL. The main
disadvantage is also connected with the universality: the determinization
step does not employ the fact that the NBA represents only an LTL definable
property. One can easily observe that \texttt{ltl2dstar} produces
unnecessarily large automata, especially for formulae with more fairness
subformulae.

The second approach is to avoid Safra's construction.
% In \cite{KV05,KPV06} the synthesis problem is reduced to emptiness 
% of nondeterministic B\"uchi tree automata.  
As probabilistic model-checkers deal with linear arithmetic, 
they do not profit 
% from tree automata approach or 
% from symbolic representation methods
from symbolically represented deterministic automata
of~\cite{PPS06,MorgensternS08}.
% Hence, deterministic $\omega$-automata are still needed for
% probabilistic model checking. 
A few translations of some simple LTL fragments to DBA have been
suggested, for example \cite{AT04}.  
Recently, a translation of
a~significantly larger LTL fragment to DRA has been introduced
in~\cite{KE12} and subsequently implemented in the tool
Rabinizer~\cite{GKE12}. The algorithm builds a~\emph{generalized
  deterministic Rabin automata (GDRA)} directly from a formula. A~DRA is
then produced by a degeneralization procedure. Rabinizer often produces
smaller automata than \texttt{ltl2dstar}. The main disadvantage is that
it works for LTL($\F,\G$) only, i.e.~the LTL fragment containing only
temporal operators \emph{eventually} ($\F$) and \emph{always}
($\G$). Authors of the translation claim that it can be extended to
a~fragment containing also the operator \emph{next} ($\X$).

In this paper, we present another Safraless translation of an LTL
fragment to DRA. 
The translation is influenced by the successful LTL
to NBA translation algorithm LTL2BA~\cite{GO01}
and it proceeds in the following three steps:
\begin{enumerate}
\item A given LTL formula $\varphi$ is translated into a \emph{very weak
    alternating co-B\"uchi automaton (VWAA)} $\mA$ as described
  in~\cite{GO01}. If $\varphi$ is an LTL($\Fs,\Gs$) formula, i.e.~any
  formula which makes use of $\F$, $\G$, and their strict variants $\Fs$ and
  $\Gs$ as the only temporal operators,
  % as well as their strict versions $\Fs$ and $\Gs$, 
  then $\mA$ satisfies an additional structural condition. We call such
  automata \emph{may/must alternating automata (MMAA)}.
\item The MMAA $\mA$ is translated into a \emph{transition-based
    generalized deterministic Rabin automaton (TGDRA)} $\mG$.  The
  construction of generalized Rabin pairs of $\mG$ is inspired
  by~\cite{KE12}.
\item Finally, % the TGDRA 
  $\mG$ is degeneralized into a (state-based) DRA $\mD$.  %in a~standard way.
\end{enumerate}

In summary, our contributions are as follows.  First, note that the
fragment LTL($\Fs,\Gs$) is strictly more expressive than
LTL($\F,\G$). Moreover, it can be shown that our translation works for a
fragment even larger than LTL($\Fs,\Gs$) but still smaller than the
whole LTL. 
% In other words, we can translate a strictly more expressive
% fragment of LTL than Rabinizer, but strictly less expressive than
% \texttt{ltl2dstar}. 
Second, the translation has a slightly better theoretical bound on the
size of produced automata comparing to \texttt{ltl2dstar}, but the same
bound as Rabinizer. Experimental results show that, for small formulae,
our translation typically produces automata of a smaller or equal size
as the other two translators.  However, for parametrized formulae, it
often produces automata that are significantly smaller. Third, we note
that our TGDRA are much smaller than the (state-based) GDRA
of~Rabinizer~\cite{GKE12}. We conjecture that algorithms for model checking of
probabilistic system, e.g.~those in PRISM~\cite{KNP11}, can be adapted
to work with TGDRA as they are adapted to work with GDRA~\cite{CGK13}.

\section{Preliminaries}\label{sec:prelim}
This section recalls the notion of linear temporal logic (LTL)~\cite{Pnu77} 
and describes the $\omega$-automata used in the following.

  %{{{ LTL

\subsubsection{Linear Temporal Logic (LTL)}

The syntax of LTL is defined by
\[
\varphi~::=~\true~\mid~a~\mid~\neg\varphi~\mid~\varphi\vee\varphi~\mid~
\varphi\wedge\varphi~\mid~\X\varphi~\mid~ \varphi\U\varphi\textrm{,}
\]
where $\true$ stands for \emph{true}, $a$ ranges over a countable set $\AP$
of \emph{atomic propositions}, $\X$ and $\uU$ are temporal operators called
\emph{next} and \emph{until}, respectively. An \emph{alphabet} is a finite
set $\Sigma=2^{\AP'}$, where $\AP'$ is a finite subset of $\AP$. An
\emph{$\omega$-word} (or simply a \emph{word}) over $\Sigma$ is an infinite
sequence of letters $u=u_0u_1u_2\ldots\in\Sigma^\omega$. By $u\suf{i}$ we
denote the suffix $u\suf{i}=u_iu_{i+1}\ldots$.
%
% PRO Ev. ZKRACENI TEXTU: (drasticke zkraceni textu)
% *** zakomentovana 1.varianta def.semantiky LTL ***
%
% We refer the reader to \cite{Pnu77} for the semantics of LTL. 
% A word $u$ \emph{satisfies} a formula $\varphi$ is denoted 
% by $u\models\varphi$.

We inductively define when a word $u$ \emph{satisfies} a formula $\varphi$,
written $u\models\varphi$, as follows.
%%%
% PRO Ev. ZKRACENI TEXTU: (4 radky namisto 7 radku)
% *** zakomentovana 2.varianta def.semantiky LTL ***
%
% \[
% \begin{array}{l@{\hspace{.5em}}c@{\hspace{.5em}}l@{\hspace{5em}}
%               l@{\hspace{.5em}}c@{\hspace{.5em}}l}  
%   u\models\true  &              & 
%   &u\models a &\mathrm{ iff } & a\in u_0  
%   \\
%   u\models\varphi_1\vee\varphi_2 &\mathrm{iff} 
%   &u\models\varphi_1 \:\mathrm{or}\: u\models\varphi_2
%     &u\models\neg\varphi &\mathrm{iff} & u\not\models\varphi
%   \\
%   u\models\varphi_1\wedge\varphi_2  &\mathrm{iff} 
%   &u\models\varphi_1 \:\mathrm{and}\: u\models\varphi_2
%     &u\models\X\varphi &\mathrm{iff} & u\suf{1}\models\varphi
%   \\
%   u\models\varphi_1\U\varphi_2 &\mathrm{iff}
%   &\multicolumn{4}{l}{\exists i\ge 0\,.\,(\,u\suf{i}\models\varphi_2
%    \:\mathrm{and}\:  \forall\, 0\leq j<i\,.~u\suf{j}\models\varphi_1\,)}
% \end{array}
% \]
% %%%
\begin{center}
\begin{tabbing}
  \hspace*{1em} \= $u\models\true$\\
  \> $u\models a$ \hspace*{2.7em} \= iff~ \= $a\in u_0$\\
  \> $u\models\neg\varphi$ \> iff \> $u\not\models\varphi$\\
  \> $u\models\varphi_1\vee\varphi_2$ \> iff \>
  $u\models\varphi_1$ or $u\models\varphi_2$\\
  \> $u\models\varphi_1\wedge\varphi_2$ \> iff \>
  $u\models\varphi_1$ and $u\models\varphi_2$\\
  \> $u\models\X\varphi$ \> iff \> $u\suf{1}\models\varphi$\\
  \> $u\models\varphi_1\U\varphi_2$ \> iff \>
     $\exists i\ge 0\,.\,(\,u\suf{i}\models\varphi_2$ and
     $\forall\, 0\leq j<i\,.~u\suf{j}\models\varphi_1\,)$
\end{tabbing}
\end{center}

%  Two formulae $\varphi,\psi$ are \emph{equivalent},
% written $\varphi\equiv\psi$, if for each alphabet $\Sigma$ and each
% $u\in\Sigma^\omega$ it holds $u\models\varphi\iff u\models\psi$.
Given an alphabet $\Sigma$, a formula $\varphi$ defines the language
$L^\Sigma(\varphi)=\{u\in\Sigma^\omega\mid u\models\varphi\}$. We 
write $L(\varphi)$ instead of $L^{2^{\AP(\varphi)}}(\varphi)$, where
$\AP(\varphi)$ denotes the set of atomic propositions occurring in the
formula $\varphi$.

We define derived unary temporal operators \emph{eventually}
($\F$), \emph{always} ($\G$), \emph{strict eventually} ($\Fs$), and
\emph{strict always} ($\Gs$) by the
following equivalences: $\F\varphi\equiv\true\U\varphi$,
$\G\varphi\equiv\neg\F\neg\varphi$, $\Fs\varphi\equiv\X\F\varphi$, and
$\Gs\varphi\equiv\X\G\varphi$.

LTL($\F,\G$) denotes the LTL fragment consisting of formulae built with
temporal operators $\F$ and $\G$ only.  The fragment build with temporal
operators $\Fs$, $\Gs$, $\F$ and $\G$ is denoted by LTL($\Fs,\Gs$) as
$\F\varphi$ and $\G\varphi$ can be seen as abbreviations for
$\varphi\lor\Fs\varphi$ and $\varphi\land\Gs\varphi$, respectively.  Note
that LTL($\Fs,\Gs$) is strictly more expressive than LTL($\F,\G$) as
formulae $\Fs a$ and $\Gs a$ cannot be equivalently expressed in
LTL($\F,\G$). 

An LTL formula is in \emph{positive normal form} if no operator occurs in
the scope of any negation.  Each LTL($\Fs,\Gs$) formula can be transformed
to this form using De Morgan's laws for $\land$ and $\lor$ and the
equivalences $\neg\Fs\psi\equiv\Gs\neg\psi$, $\neg\Gs\psi\equiv\Fs\neg\psi$,
$\neg\F\psi\equiv\G\neg\psi$, and $\neg\G\psi\equiv\F\neg\psi$. We say that
a formula is \emph{temporal} if its topmost operator is neither conjunction,
nor disjunction (note that $a$ and $\neg a$ are also temporal formulae).

  %}}}
  %{{{ Deterministic automata 

\subsubsection{Deterministic Rabin Automata and Their Generalization}%\label{def:DA}

% \kickoff{DA SE VYPUSTIT, KDYZ RABINA ODCITUJEM JINDE. In this section we
%   briefly define a deterministic version of \emph{Rabin
%     automata}~\cite{Rab69} and a new generalization of this formalism called
%   \emph{transition-based generalized deterministic Rabin automata} (TGDRA).}

A \emph{semiautomaton} is a tuple $\mT=(S,\Sigma,\delta,s_I)$, where $S$ is
a finite set of \emph{states}, $\Sigma$ is an alphabet, $s_I\in S$ is the
\emph{initial state}, and $\delta\subseteq S\times\Sigma\times S$ is a
deterministic \emph{transition relation}, i.e.~for each state $s\in S$ and
each $\alpha\in\Sigma$, there is at most one state $s'$ such that
$(s,\alpha,s')\in\delta$. A triple $(s,\alpha,s')\in\delta$ is called a
\emph{transition} from $s$ to $s'$ labelled by $\alpha$, or an
$\alpha$-transition of $s$ leading to $s'$.  In illustrations, all
transitions with the same source state and the same target state are usually
depicted by a single edge labelled by a propositional formula $\psi$ over
$\AP$ representing the corresponding transition labels (e.g.~given
$\Sigma=2^{\{a,b\}}$, the formula $\psi=a\vee b$ represents labels
$\{a\},\{a,b\},\{b\}$).

A \emph{run} of a semiautomaton $\mT$ over a word $u=
u_0u_1\ldots\in\Sigma^\omega$ is an infinite sequence $\sigma=
(s_0,u_0,s_1)(s_1,u_1,s_2)\ldots\in\delta^\omega$ of transitions such that
$s_0=s_I$. By $\Inf_t(\sigma)$ (resp.~$\Inf_s(\sigma)$) we denote the set of
transitions (resp.~states) occurring infinitely often in $\sigma$. For each
word $u\in\Sigma^\omega$, a semiautomaton has at most one run over $u$
denoted by $\sigma(u)$.

A \emph{deterministic Rabin automaton} (DRA) is a tuple
$\mD=(S,\Sigma,\delta,s_I,\mathcal{R})$, where $(S,\Sigma,\delta,s_I)$ is a
semiautomaton and $\mathcal{R} \subseteq 2^{S} \times 2^{S}$ is a finite set
of \emph{Rabin pairs}.
% \kickoff{TOHLE POTREBUJEM? Each Rabin pair $(K,L)\in\mathcal{R}$ consists
% of the set $K$ of \emph{rejecting} states
% and the set $L$ of \emph{accepting} states.}
Runs of $\mD$ are runs of the semiautomaton.  A run $\sigma$ \emph{satisfies}
a Rabin pair $(K,L)\in\mathcal{R}$ if $\Inf_s(\sigma)\cap K=\emptyset$ and
$\Inf_s(\sigma)\cap L\neq\emptyset$.  A run is \emph{accepting} if it
satisfies some Rabin pair of $\mathcal{R}$. The language of $\mD$ is the set
$L(\mD)$ of all words $u\in\Sigma^\omega$ such that $\sigma(u)$ is
accepting.

% \kickoff{TOHLE SE DA SMAZAT? Each Rabin pair $(K,L)\in\mathcal{R}$ consists
%   of the set $K$ of \emph{rejecting} states and the set $L$ of
%   \emph{accepting} states.}  A run of a DRA is defined as a run of the
%   corresponding semiautomaton.  A run $\sigma$ of a DRA is \emph{accepting}
%   if there exists a Rabin pair $(K,L)\in\mathcal{R}$ such that
%   $\Inf_s(\sigma)\cap K=\emptyset$ and $\Inf_s(\sigma)\cap L\neq\emptyset$,
%   i.e.~$\sigma$ contains only finitely many occurrences of rejecting states
%   and infinitely many occurrences of some accepting state.}

A \emph{transition-based generalized deterministic Rabin automaton} (TGDRA)
is a tuple $\mG=(S,\Sigma,\delta,s_I,\mathcal{GR})$, where
$(S,\Sigma,\delta,s_I)$ is a semiautomaton and $\mathcal{GR} \subseteq
2^{\delta} \times 2^{2^{\delta}}$ is a finite set of \emph{generalized Rabin
  pairs}.  
% \kickoff{TOHLE POTREBUJEM? The first component of a Rabin pair is a set of
% \emph{rejecting transitions} and the second component is a set of
% \emph{accepting sets of transitions}.}
Runs of $\mG$ are runs of the semiautomaton.  A run $\sigma$
\emph{satisfies} a generalized Rabin pair $(K,\{L_j\}_{j\in
  J})\in\mathcal{GR}$ if $\Inf_t(\sigma)\cap K=\emptyset$ and, for each
$j\in J$, $\Inf_t(\sigma)\cap L_j\neq\emptyset$.  A run is \emph{accepting}
if it satisfies some generalized Rabin pair of $\mathcal{GR}$.  The language
of $\mG$ is the set $L(\mG)$ of all words $u\in\Sigma^\omega$ such that
$\sigma(u)$ is accepting.

A generalization of DRA called \emph{generalized deterministic Rabin
  automata} (GDRA) has been considered in~\cite{KE12,GKE12}. The accepting
condition of GDRA is a positive Boolean combination (in disjunctive normal
form) of Rabin pairs. A run $\sigma$ is accepting if $\sigma$ satisfies this
condition.
% (a Rabin pair $(K,L)\in 2^S\times 2^S$ is satisfied by $\sigma$ if
% $\Inf_s(\sigma)\cap K=\emptyset$ and $\Inf_s(\sigma)\cap L\neq\emptyset$).

% A generalized Rabin pair is satisfied by a run $\sigma$ of $\mD$ if and only
% if $\sigma$ uses the rejecting transitions only finitely often and uses
% infinitely often transitions from each accepting set of the pair.
% \kickoff{NEBO RADSI: The Rabin pair $(K_i,\{L_i^j\}_{j \in J_i})$ is
%   satisfied by a run $\sigma$ of $\mD$ if and only if $\sigma$ uses finitely
%   often transitions from $K_i$ and infinitely often transitions from each of
%   $L_i^j$.}

% \kickoff{Zduraznit rozdily s definici generalized Rabin automata by
%   K\v{r}et\'{i}nsk\'{y} and Esparza~\cite{KE12}?  Opposed to their
%   definition we define transition based version of generalized RA in order
%   to obtain smaller state space. Moreover our definition of the acceptance
%   condition is different. We propose the condition to be defined as a set of
%   generalized pairs instead of boolean combination over states as used in
%   \cite{KE12}. We believe that such definition is more compact and the
%   condition can be represented by simple data structures.}

  %}}}
  %{{{ VWAA + MMAA

\subsubsection{Very Weak Alternating Automata and Their Subclass}
\label{def:MMAA}

A \emph{very weak alternating co-B\"{u}chi automaton} (VWAA) $\mA$
% \footnote{The formalism is also known as \emph{1-weak} or \emph{linear
% alternating co-B\"{u}chi automata}.}
is a tuple $(S,\Sigma,\delta,I,F)$, where $S$ is a finite set of
\emph{states}, subsets $c\subseteq S$ are called \emph{configurations},
$\Sigma$ is an \emph{alphabet}, $\delta \subseteq S \times \Sigma \times
2^S$ is an \emph{alternating transition relation}, $I \subseteq 2^S$ is a
non-empty set of \emph{initial configurations}, $F\subseteq S$ is a set of
\emph{co-B\"{u}chi accepting states}, and there exists a partial order on
$S$ such that, for every transition $(s,\alpha,c)\in\delta$, all the states
of $c$ are lower or equal to $s$.

A triple $(s,\alpha,c)\in\delta$ is called a transition from $s$ to $c$
labelled by $\alpha$, or an $\alpha$-transition of $s$.  We say that $s$ is
the \emph{source state} and $c$ the \emph{target configuration} of the
transition.  A transition is \emph{looping} if the target configuration
contains the source state, i.e.~$s\in c$.  A transition is called a
\emph{selfloop} if its target configuration contains the source state only,
i.e.~$c=\{s\}$.

%{{{ Figure of an automaton and its run

\begin{figure}[!t]
  \centering
  \subfigure[]{\label{fig:alt}
    \begin{tikzpicture}[auto,initial where=above,initial text=,semithick,>=stealth]
      %%% NAME OF AUTOMATON %%%
      %\node at (0.65,3.5) {$\mA_{\G(\F\s a \land \F\s b) \lor \G b}$};
      % \draw[step=1,gray,very thin] (0,0) grid (4,4);

      %%% STATES %%%
      \node[state,initial]   (s0) at (0,2){$\G\psi$};
      \node[state,accepting](s1) at (0,0){$\F a$};
      \node[state,accepting](s2) at (1.3,0){$\F b$};
      \node[state,initial]   (s3) at (1.3,2){$\G b$};
      \node[] (placeholder) at (0.65,-1) {\strut};

      \path[->] %transitions of G
      (s0) edge node [right,pos=0.65] {$\true$} (s1) 
      (s0) edge [out=-90,in=90,looseness=1.5] node {} (s2)
      (s0) edge [looseness=8,out=-90,in=-120,overlay] node [] {} (s0);
      \path[->] %transition of Gb
      (s3) edge [looseness=8,out=-75,in=-45,overlay] node [right,pos=.63] {$b$} (s3);
      \path[->] %transitions of may states
      (s1) edge [looseness=8,out=-105,in=-135,overlay] node [left,pos=.63] {$\true$} (s1)
      (s1) edge [] node [right] {$a$} +(0,-1)	 	 
      (s2) edge [looseness=8,out=-75,in=-45,overlay] node [right,pos=.63] {$\true$} (s2)
      (s2) edge [] node [left] {$b$} +(0,-1);
    \end{tikzpicture}
  }~
  \subfigure[]{\label{fig:vwaarun}
    \def\xbase{0.7} 
    \def\yG{3}
    \def\yFa{2}
    \def\yFb{1}
    \def\yGb{0}
    \def\yl{\yG+0.4}
    \def\ylast{\yGb}
    \def\length{8}

    \begin{tikzpicture}[inner sep=0pt,>=stealth,
      dot/.style={fill=black,circle,minimum size=5pt},xscale=0.85]
      % Labels of nodes 
      \node[state,run] at (0,\yG) {$\G\psi$};
      \node[state,run,accepting] at (0,\yFa) {$\F a$};
      \node[state,run,accepting] at (0,\yFb) {$\F b$};
      \node[state,run] at (0,\yGb) {$\G b$};
      \node[] (placeholder) at (0.65,-.5) {\strut};

      % Labels of edges
      \node at (0.5+\xbase, \yG+0.5) {$\{a\}$};
      \node at (1.5+\xbase, \yG+0.5) {$\emptyset$};
      \node at (2.5+\xbase, \yG+0.5) {$\{b\}$};
      \node at (3.5+\xbase, \yG+0.5) {$\{a,b\}$};
      \node at (4.5+\xbase, \yG+0.5) {$\{a\}$};
      \node at (5.5+\xbase, \yG+0.5) {$\emptyset$};
      \node at (6.5+\xbase, \yG+0.5) {$\{b\}$};
      \node at (7.5+\xbase, \yG+0.5) {$\{a,b\}$};
%      \node at (8.5+\xbase, \yG+0.5) {$\ldots$};

      % nodes for
      \foreach \x in {0,...,\length} \node[dot] (G\x) at (\xbase+\x,\yG) {};
      \foreach \x in {1,...,\length} \node[dot] (Fa\x) at (\xbase+\x,\yFa) {};
      \foreach \x in {1,...,\length} \node[dot] (Fb\x) at (\xbase+\x,\yFb) {};

      % Edges and nodes of G
      \foreach \x in {0,...,7} {
	\node[dot] (G\x+1) at (\xbase+\x+1,\yG) {};
	\node[dot] (Fa\x+1) at (\xbase+\x+1,\yFa) {};
	\node[dot] (Fb\x+1) at (\xbase+\x+1,\yFb) {};	
	\path[->] (G\x) edge (G\x+1);
	\path[->] (G\x) edge (Fa\x+1);
	\path[->] (G\x) edge (Fb\x+1);
      }

      % True-loops edges
      \path[->] (Fb1) edge (Fb2)
                (Fa1) edge (Fa2)
                (Fa2) edge (Fa3)
                (Fb4) edge (Fb5)	
                (Fb5) edge (Fb6)
                (Fa5) edge (Fa6)	
                (Fa6) edge (Fa7);
                
      % edges to empty configuration
      \path[->] (Fb2) edge +(0.3,-0.3)
      		    (Fa3) edge +(0.3,-0.3)
      		    (Fb3) edge +(0.3,-0.3)
      		    (Fa4) edge +(0.3,-0.3)
				(Fb7) edge +(0.3,-0.3)
				(Fa7) edge +(0.3,-0.3)
				(Fb6) edge +(0.3,-0.3);
      % continues to infinity
      \node at (\xbase+\length+0.5,\yFb+0.5) {$\cdots$};
		  
      % Slices    
      \foreach \x in {0,...,\length} {
        \draw[dotted] (\xbase+\x,\ylast-0.3) to (\xbase+\x,\yG+0.3);
        \node at (\xbase+\x,\yG+0.3) {\tiny{\x}};
      }

      % Names of multi-transitions    
      \foreach \x in {0,...,7} \node at (\xbase+\x+0.5,\ylast-0.5) {\small{$T_{\x}$}};
    \end{tikzpicture}
  }
  \caption{\subref{fig:alt} A VWAA (and also
    MMAA) %$\mA_{\G(\F\s a \land \F\s b) \lor \G b}$
    corresponding to formula $\G\psi\lor\G b$, where $\psi=\Fs a\land\Fs
    b$. %\kickoff{ZAVEST OZNACENI $\mA_{\G(\F\s a \land \F\s b) \lor \G b}$?}
    \subref{fig:vwaarun} An accepting run of the automaton
    % $\mA_{\exampleformula}$
    over $(\{a\}\emptyset\{b\}\{a,b\})^\omega$.}
\end{figure}

%}}}

Figure~\ref{fig:alt} shows a VWAA that accepts the
language described by the formula $\G(\Fs a \land \Fs b) \lor \G b$.
Transitions are depicted by branching edges.  If a target configuration
is empty, the corresponding edge leads to an empty space.  We often depict all
transitions with the same source state and the same target configuration by
a single edge (as for semiautomata).  Each initial configuration is
represented by a possibly branching unlabelled edge leading from an empty
space to the states of the configuration.  Co-B\"{u}chi accepting states are
double circled.

% It has four states. Each of $\G$ and $\G b$ forms one initial configuration
% (denoted by incoming edge from empty space) and both $\F a$ and $\F b$ are
% co-B\"{u}chi accepting. We use boolean combinations over atomic propositions
% as labels of transitions to merge several transitions into one. Transition
% labelled by $\varphi$ represents transitions that differs only in labels and
% all the labels satisfy $\varphi$. The edge from $\F a$ to empty space
% represents two transitions: $(\F a, \{a\}, \emptyset)$ and $(\F a, \{a,b\},
% \emptyset)$.

A \emph{multitransition} $T$ with a label $\alpha$ is a set of transitions
with the same label and such that the source states of the transitions are
pairwise different.  A \emph{source configuration} of $T$, denoted by
$\dom(T)$, is the set of source states of transitions in $T$.  A
\emph{target configuration} of $T$, denoted by $\range(T)$, is the union of
target configurations of transitions in $T$.  We define a
\emph{multitransition relation} $\Delta \subseteq 2^S \times \Sigma \times
2^S$ as $$\Delta=\{(\dom(T), \alpha, \range(T)) \mid \text{there exists a
  multitransition } T \text{ with label }\alpha\}.$$

%%%%%% jestli to chceme pouzit, tak je potreba opravit \Fs na \F.
% Look at our running example in figure~\ref{fig:alt}. One possible
% multitransition of $\{\G\psi, \F\s a\}$ is $T = \{(\G\psi, \{a\}, \{\G\psi,
% \F\s a, \F\s b\} ), (\F\s a, \{a\}, \emptyset)\}$. Then $\dom(T) = \{\G\psi,
% \F\s a\}$ and $\range(T) = \{\G\psi, \F\s a, \F\s b\}$. Therefore
% $(\{\G\psi, \F\s a\},\{a\},\{\G\psi, \F\s a, \F\s b\}) \in \Delta$.  Another
% multitransition with the same source and target configurations is $T_1 =
% \{(\G\psi, \{a\}, \{\G\psi, \F\s b\, \F\s a\} ), (\F\s a, \{a\}, \F\s a)\}$.

A \emph{run} $\rho$ of a VWAA $\mA$ over a word $w = w_0w_1\ldots \in
\Sigma^\omega$ is an infinite sequence $\rho=T_0T_1\ldots$ of multitransitions of
$\mA$ such that $\dom(T_0)$ is an initial configuration of $\mA$ and, for
each $i\ge 0$, $T_i$ is labelled by $w_i$ and $\range(T_i) =
\dom(T_{i+1})$. 

A run can be represented as a directed acyclic graph (DAG).  For example,
the DAG of Figure~\ref{fig:vwaarun} represents a run of the VWAA of
Figure~\ref{fig:alt}.  The dotted lines divide the DAG into segments
corresponding to multitransitions.  Each transition of a multitransition is
represented by edges leading across the corresponding segment from the
starting state to states of the target configuration. As our alternating
automata are very weak, we can order the states in a way that all edges
in any DAG go only to the same or a lower row.

An accepting run corresponds to a DAG where each branch contains only
finitely many states from $F$.  Formally, the run $\rho$ is \emph{accepting}
if it has no suffix where, for some co-B\"{u}chi accepting state $f\in F$,
each multitransition contains a looping transition from $f$. The language of
$\mA$ is the set $L(\mA)=\{w\in\Sigma^\omega\mid\mA\text{ has an accepting
  run of over }w\}$.  By $\Inf_s(\rho)$ we denote the set of states that
occur in $\dom(T_i)$ for infinitely many indices $i$.

% Nodes of the DAG are divided into \emph{slices} (depicted by dotted lines).
% Nodes are labelled by states (all nodes labelled by a state in the leftmost
% column are in the corresponding raw) and edges are labelled by sets of
% atomic propositions (depicted above). Labels of nodes of $i$-th slice give
% the set of states active before $T_i$, i.e. it is equal to $\dom(T_i) =
% \range(T_{i-1})$ for $i > 0$. The labels of nodes of slice $0$ form some
% initial configuration ($\{\G\psi\}$ in case of figure \ref{fig:vwaarun}).

% All edges leaving one node represent one transition. Transitions of all
% nodes of $i$-th slice represent the multitransition $T_i$ (below the
% edges). Note that the DAG can contain only edges with non-increasing
% labelling for very weak alternating automata.

\begin{definition}\label{def:mmaa}
  A \emph{may/must alternating automaton} (MMAA)
  is a VWAA where each state fits into one of the following three
  categories:
  \begin{enumerate}
  \item \emph{May-states} -- states with a selfloop for each
    $\alpha\in\Sigma$.  A run that enters such a state \emph{may} wait in
    the state for an arbitrary number of steps.
  \item \emph{Must-states} -- every transition of a must-state is looping. A
    run that enters such a state can never leave it. In other words, the run
    \emph{must} stay there.
  \item \emph{Loopless states} -- states that have no looping transitions
    and no predecessors. They can appear only in initial configurations (or
    they are unreachable).
  \end{enumerate}
\end{definition}

The automaton of Figure~\ref{fig:alt} is an MMAA with must-states
$\G\psi,\G b$ and may-states $\F a,\F b$.

We always assume that the set $F$ % of co-B\"uchi accepting states
of an MMAA coincides with the set of all may-states of the automaton.  This
assumption is justified by the following observations:
\begin{itemize}
\item There are no looping transitions of loopless states. Hence,
  removing all loopless states from $F$ has no effect on acceptance of
  any run.
\item All transitions leading from must-states are looping. Hence, if a
  run contains a must-state that is in $F$, then the run is
  non-accepting.  Removing all must-states in $F$ together with their
  adjacent transitions from an MMAA has no effect on its accepting runs.
\item Every may-state has selfloops for all $\alpha\in\Sigma$. If such a
  state is not in $F$, we can always apply these selfloops without violating
  acceptance of any run. We can also remove these states from all the target
  configurations of all transitions of an MMAA without affecting its
  language.
\end{itemize}

%}}}

%}}}
%{{{ LTL(Fs,Gs)->MMAA 

\section{Translation of LTL($\Fs,\Gs$) to MMAA}\label{sec:corr}

We present the standard translation of LTL to VWAA~\cite{GO01} restricted to
the fragment LTL($\Fs,\Gs$).  In this section, we treat the transition
relation $\delta\subseteq S\times\Sigma\times2^S$ of a VWAA as a function
$\delta:S\times\Sigma\rightarrow2^{2^S}$, where $c\in\delta(s,\alpha)$ means
$(s,\alpha,c)\in\delta$.  Further, we consider $\G\psi$ and $\F\psi$ to be
subformulae of $\Gs\psi$ and $\Fs\psi$, respectively.  This is justified by
equivalences $\Gs\psi\equiv\X\G\psi$ and $\Fs\psi\equiv\X\F\psi$.  Recall
that a formula is called \emph{temporal} if its topmost operator is neither
conjunction, nor disjunction (note that $a$ and $\neg a$ are also temporal
formulae).

Let $\varphi$ be an $\LTL(\Fs,\Gs)$ formula in positive normal form.  An
equivalent VWAA is constructed as $\mA_\varphi=(Q,\Sigma,\delta,I,F)$, where
\begin{itemize}
\item $Q$ is the set of temporal subformulae of $\varphi$,
  % enriched with a formula $\G\psi$ for each subformula $\Gs\psi$ and with
  % a formula $\F\psi$ for each subformula $\Fs\psi$,
\item $\Sigma=2^{\AP(\varphi)}$,
\item $\delta$ is defined as 
  \[
  \begin{array}{rclp{8ex}rcl}
    \delta(\true,\alpha) & = & \{\emptyset\} 
    && \delta(a,\alpha) & = &
    \{\emptyset\} \textrm{ if }a\in\alpha\textrm{, }\emptyset\textrm{ otherwise} 
    \\
    \delta(\neg\true,\alpha) & = & \emptyset 
    && \delta(\neg a,\alpha) & = & 
    \{\emptyset\} \textrm{ if }a\not\in\alpha\textrm{, }\emptyset\textrm{ otherwise} 
    \\
    \delta(\Gs\psi,\alpha) & = & \{\{\G\psi\}\} 
    && \delta(\G\psi,\alpha) & = & \{c\cup\{\G\psi\}\mid c\in\overline{\delta}(\psi,\alpha)\}
    \\
    \delta(\Fs\psi,\alpha) & = & \{\{\F\psi\}\} 
    && \delta(\F\psi,\alpha) & = &
    \{\{\F\psi\}\}\cup\overline{\delta}(\psi,\alpha)\textrm{, where}
    \\
  \end{array}
  \]
  \[
  \begin{array}{rcl}
    \overline{\delta}(\psi,\alpha) & = & 
    \delta(\psi,\alpha)\ \textrm{if }\psi\textrm{ is a temporal formula} \\
    \overline{\delta}(\psi_1\lor\psi_2,\alpha) & = & \overline{\delta}(\psi_1,\alpha)\cup\overline{\delta}(\psi_2,\alpha) \\
    \overline{\delta}(\psi_1\land\psi_2,\alpha) & = & 
    \{c_1\cup c_2\mid c_1\in\overline{\delta}(\psi_1,\alpha)\textrm{ and }c_2\in\overline{\delta}(\psi_2,\alpha)\}\textrm, \\
  \end{array}
  \]
\item $I=\overline{\varphi}$ where $\overline{\varphi}$ is defined
  as 
  \[
  \begin{array}{rcl}
    \overline{\psi} & = & \{\{\psi\}\}\textrm{ if }\psi\textrm{ is a temporal formula} \\
    \overline{\psi_1 \vee \psi_2} & = & \overline{\psi_1}\cup\overline{\psi_2} \\
    \overline{\psi_1\land\psi_2} & = &
    \{O_1\cup O_2\mid O_1\in\overline{\psi_1}\textrm{ and }O_2\in\overline{\psi_2}\} \textrm{, and}
  \end{array}
  \]
\item $F\subseteq Q$ is the set of all subformulae of the form $\F\psi$ in $Q$.
\end{itemize}

Using the partial order ``is a subformula of'' on states, one can easily
prove that $\mA_\varphi$ is a VWAA.  Moreover, all the states of the form
$\G\psi$ are must-states and all the states of the form $\F\psi$ are may-states.
States of other forms are loopless and they are unreachable unless they
appear in $I$. Hence, the constructed automaton is also an MMAA.
Figure~\ref{fig:alt} shows an MMAA produced by the translation of formula
$\G(\Fs a \land \Fs b) \lor \G b$.

%\todo{Dame sem nejaky theorem o ekvivalenci $\varphi$ a $\mA_\varphi$?}
% \begin{theorem}
%   For each formula $\varphi\in\LTL(\F_s,\Gs)$, we can construct an MMAA
%   $\mA_\varphi$ such that $L(\varphi)=L(\mA_\varphi)$.  
% \end{theorem}

In fact, MMAA and LTL($\Fs,\Gs$) are expressively equivalent.  The reverse
translation can be found in Appendix~\ref{sec:mmaa2ltl}.

%}}}
%{{{ MMAA->TGDRA

\section{Translation of MMAA to TGDRA}
\label{sec:MMAA2DRA}
In this section we present a translation of an MMAA $\mA =
(S,\Sigma,\delta\A,I,F)$ with multitransition relation $\Delta\A$ into an
equivalent TGDRA $\mG$. At first we build a semiautomaton $\mT$ by a double
powerset construction (performing dealternation and determinization of the
MMAA). Then we describe the transition based generalized Rabin acceptance
condition $\mGR$ of $\mG$.

  %{{{ Semiautomaton T

\subsection{Semiautomaton $\mT$}
\label{sec:ts}
The idea of our seminautomaton construction is straightforward: a run
$\sigma(w)$ of the semiautomaton $\mT$ tracks all runs of $\mA$ over $w$.
More precisely, the state of $\mT$ reached after reading a finite input
consists of all possible configurations in which $\mA$ can be after reading
the same input.  Hence, states of the semiautomaton are sets of
configurations of $\mA$ and we call them \emph{macrostates}.  We use
$f,s,s_1,s_2,\ldots$ to denote states of $\mA$ ($f$ stands for an accepting
state of $F$), $c,c_1,c_2,\ldots$ to denote configurations of $\mA$, and
$m,m_1,m_2,\ldots$ to denote macro\-states of $\mT$.  Further, we use
$t,t_1,t_2\ldots$ to denote the transitions of $\mA$, $T,T_0,T_1\ldots$ to
denote multitransitions of $\mA$, and $r,r_1,r_2\ldots$ to denote the
transitions of $\mT$, which are called \emph{macrotransitions} hereafter.
% Before we give the formal definition of the semiautomaton we recall that
% $(c_1,\alpha,c_2)\in\Delta\A$ means that some multitransition $T$ with
% label $\alpha$ such that $\dom(T)=c_1$ and $\range(T)=c_2$ exists.

Formally, we define the \emph{semiautomaton $\mT = (M,\Sigma,\delta\T,m_{I})$} 
for $\mA$ as follows:
\begin{itemize}
\item $M \subseteq 2^{2^{S}}$ is the set \emph{macrostates}, restricted to
  those reachable from the initial macrostate $m_I$ by $\delta\T$,
%\item $\delta\T: M\times\Sigma\times M$ is the \emph{deterministic
%    transition relation} and
\item $(m_1,\alpha,m_2) \in \delta\T$ iff $m_2 = \bigcup_{c\in m_1} \{c'
  \mid (c, \alpha, c') \in \Delta\A \}$, i.e.~for each $m_1\in M$ and
  $\alpha\in\Sigma$, there is a single macrotransition
  $(m_1,\alpha,m_2)\in\delta\T$, where $m_2$ consists of target
  configurations of all $\alpha$-multitransitions leading from
  configurations in $m_1$, and
\item $m_{I}= I$ is the \emph{initial macrostate}.
\end{itemize}
%
% \kickoff{The transition relation $\delta\T$ follows and
% summarizes the semantics of the alternating transition relation
% $\delta\A$. The initial macrostate contains the initial configurations of
% $\mA$.}
%
% In other words, an $\alpha$-successor $m_2$ of a macrostate $m_1$ consists
% of all $\alpha$-successors of all configurations in $m_1$.

% \kickoff{TOHLE JE NA NECO? MNE TO PRIPADA JASNY. Note that for each
%   configuration $c'$ of a macrostate $m_2$ and each $\alpha \in \Sigma$ such
%   that $m_2$ has some $\alpha$-predecessor $m_1$, there is some
%   configuration $c$ in $m_1$ such that $(c, \alpha, c') \in \Delta$.}

Figure~\ref{fig:ts} depicts the semiautomaton $\mT$ for the MMAA of
Figure~\ref{fig:alt}. Each row in a macrostate represents one configuration.

\begin{figure}[t]
  \centering
    \begin{tikzpicture}[auto,initial where=left,initial text=,
      semithick,>=stealth,trim left=(init.west),trim right=(gp.east)]
%      \tikzstyle{state}=[shape=rectangle,rounded corners=15pt,thick,draw,
%        minimum width=2.1cm, minimum height=1.1cm,inner sep=2pt,align=center]
      \tikzstyle{state}=[shape=rectangle with rounded corners,
        rectangle corner radius=15pt,thick,draw,
        minimum width=2.1cm, minimum height=1.1cm,inner sep=2pt,align=center]
      %macrostate definitionsnon-accepting run $\rho_1$
      \node[state,initial] (init) at (0,0) {$\{\G\psi\}$ \\ $\{\G b\} $};
      \node[state] (gpgb) at (3,0) {$\{\G\psi, \F a, \F b\}$ \\ $\{\G b\}$};
      \node[state] (gp) at (6,0) {$\{\G\psi, \F a, \F b\}$};
      %transitions
      \path[->]
      (init) edge [] node [above] {$b$} (gpgb)
      (init) edge [bend right=23] node [below,overlay] {$\neg b$} (gp)	
      (gpgb) edge [loop above,looseness=6] node [above] {$b$} (gpgb)
      (gpgb) edge [] node [above] {$\neg b$} (gp)
      (gp) edge [loop above,looseness=6] node [above] {$\true$} (gp);
      % (gp) edge [loop,looseness=\loopLooseness,in=60, out=20]
      % node [above] {$\true$} (gp);
    \end{tikzpicture}
    
%    \begin{tikzpicture}[auto,initial where=left,initial text=,
%      semithick,>=stealth]
%      \tikzstyle{state}=[shape=rectangle,rounded corners=15pt,thick,draw,
%        minimum width=2.1cm, minimum height=1.1cm,inner sep=2pt,align=center]
%      %macrostate definitions$
%      \node[state,initial] (gpgb) at (3,0) {$\{\G\psi, \F a, \F b\}$ \\ $\{\G b\}$};
%      \node[state] (gp) at (6,0) {$\{\G\psi, \F a, \F b\}$};
%      %transitions
%      \path[->]
%      (gpgb) edge [loop above,looseness=6] node [above] {$b \land a$} (gpgb)
%      (gpgb) edge [loop below,looseness=6] node [below] {$b \land \neg a$} (gpgb)
%      (gpgb) edge [bend right] node [below] {$\neg b \land a$} (gp)
%      (gpgb) edge [bend left] node [above] {$\neg b \land \neg a$} (gp)
%      (gp) edge [loop above,looseness=6] node [above] {$\true$} (gp);
%      % (gp) edge [loop,looseness=\loopLooseness,in=60, out=20]
%      % node [above] {$\true$} (gp);
%    \end{tikzpicture}
  \caption{The semiautomaton $\mT$ for the MMAA of Figure \ref{fig:alt}.}
  \label{fig:ts}
\end{figure}
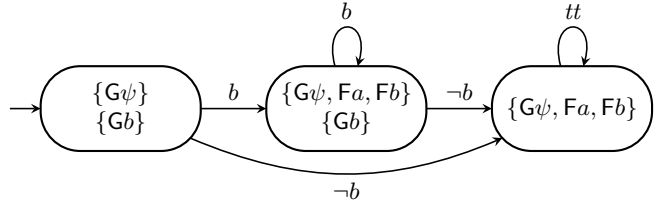

%}}}

  %}}}
  %{{{ Acceptance conditons

\subsection{Acceptance Condition $\mGR$ of the TGDRA $\mG$}
For any subset $Z\subseteq S$, $\must(Z)$ denotes the set of must-states
of $Z$. An MMAA run $\rho$ is \emph{bounded by} $Z \subseteq S$ iff $\Inf_s(\rho)
\subseteq Z$ and $\must(\Inf_s(\rho))=\must(Z)$.  For example, the run of
Figure~\ref{fig:vwaarun} is bounded by the set $\{\G\psi, \F a, \F b\}$.

For any fixed $Z \subseteq S$, we define the set $\AC_Z \subseteq 2^{S}$ of
\emph{allowed configurations} of $\mA$ and the set $\AT{Z} \subseteq
\delta\T$ of \emph{allowed macrotransitions} of $\mT$ as follows:
\begin{eqnarray*}
  \AC_Z&=&\{ c \subseteq Z \mid \must(c) = \must(Z) \}\\
  \AT{Z}&=&\{ (m_1,\alpha,m_2) \in \delta\T \mid \exists c_1 \in \AC_Z, 
  c_2 \in (m_2 \cap \AC_Z) : (c_1,\alpha,c_2) \in \Delta\A \}\footnotemark
\end{eqnarray*}
\footnotetext{A definition of $\AT{Z}$ with $c_1\in(m_1 \cap \AC_Z)$ would
  be more intuitive, but less effective.}  Clearly, a run $\rho$ of $\mA$ is
bounded by $Z$ if and only if $\rho$ has a suffix containing only
configurations of $\AC_Z$. Let $\rho$ be a run over $w$ with such a
suffix. As the semiautomaton $\mT$ tracks all runs of $\mA$ over a given
input, the run $\sigma(w)$ of $\mT$ `covers' also $\rho$. Hence, $\sigma(w)$
has a suffix where, for each macrotransition $(m_i,w_i,m_{i+1})$, there
exist configurations $c_1\in m_i\cap\AC_Z$ and $c_2\in m_{i+1}\cap\AC_Z$
satisfying $(c_1,w_i,c_2)\in\Delta\A$. In other words, $\sigma(w)$ has a
suffix containing only macrotransitions of $\AT{Z}$.  This observation is
summarized by the following lemma.
\begin{lemma}\label{lem:aux}
  If $\mA$ has a run over $w$ bounded by $Z$, then the run $\sigma(w)$ of
  $\mT$ contains a suffix of macrotransitions of $\AT{Z}$.
\end{lemma}
In fact, the other direction can be proved as well: if $\sigma(w)$ contains
a suffix of macrotransitions of $\AT{Z}$, then $\mA$ has a run over $w$
bounded by $Z$.

For each $f\in F\cap Z$, we also define the set $\AT{Z}^f$ as the set of all
macrotransitions in $\AT{Z}$ such that $\mA$ contains a non-looping
transition of $f$ with the same label and with the target configuration not
leaving $Z$:
$$\AT{Z}^f=\{(m_1,\alpha,m_2) \in \AT{Z}\mid\exists (f,\alpha,c)\in\delta_\mA:
f\not\in c, c\subseteq Z\}$$

Using the sets $\AT{Z}$ and $\AT{Z}^f$, we define one generalized Rabin pair
$\mathcal{GR}_Z$ for each subset of states $Z \subseteq S$:
\begin{equation}
\mathcal{GR}_Z = (\delta\T \smallsetminus \AT{Z}, \{ \AT{Z}^f\}_{f \in F \cap Z})
\end{equation}

\begin{lemma}\label{lem:ag}
If there is an accepting run $\rho$ of $\mA$ over $w$ then the run 
$\sigma(w)$ of $\mT$ satisfies $\mathcal{GR}_{Z}$ for $Z=\Inf_s(\rho)$.
\end{lemma}
\begin{proof}
  As $\rho$ is bounded by $Z$, Lemma~\ref{lem:aux} implies that $\sigma(w)$
  has a suffix $r_ir_{i+1}\ldots$ of macrotransitions of $\AT{Z}$. Thus
  $\Inf_t(\sigma(w)) \cap (\delta\T \smallsetminus \AT{Z}) = \emptyset$.

  As $Z=\Inf_s(\rho)$ and $\rho=T_0T_1\ldots$ is accepting, for each $f\in F
  \cap Z$, $\rho$ includes infinitely many multitransitions $T_j$ where
  $f\in\dom(T_j)$ and $T_j$ contains a non-looping transition
  $(f,w_j,c)\in\delta\A$ satisfying $f\not\in c$ and $c\subseteq Z$. Hence,
  the corresponding macrotransitions $r_j$ that are also in the mentioned
  suffix $r_ir_{i+1}\ldots$ of $\sigma(w)$ are elements of $\AT{Z}^f$.
  Therefore, $\Inf_t(\sigma(w)) \cap \AT{Z}^f \neq \emptyset$ for each $f
  \in F \cap Z$ and $\sigma(w)$ satisfies $\mGR_Z$.\qed
\end{proof}

\begin{lemma}\label{lem:ga}
If a run $\sigma(w)$ of $\mT$ satisfies $\mathcal{GR}_Z$ then there is
an \emph{accepting} run of $\mA$ over $w$ bounded by $Z$.
\end{lemma}
\begin{proof}
  Let $\sigma(w)=r_0r_1\ldots$ be a run of $\mT$ satisfying $\mGR_Z$,
  i.e.~$\sigma(w)$ has a suffix of macrotransitions of $\AT{Z}$ and
  $\sigma(w)$ contains infinitely many macrotransitions of $\AT{Z}^f$ for
  each $f \in F \cap Z$.  Let $r_i=(m_i,w_i,m_{i+1})$ be the first
  macrotransition of the suffix.  The definition of $\AT{Z}$ implies that
  there is a configuration $c\in m_{i+1}\cap\AC_Z$.  The construction of
  $\mT$ guarantees that there exists a sequence of multitransitions of $\mA$
  leading to the configuration $c$. More precisely, there is a sequence
  $T_0T_1\ldots T_i$ such that $\dom(T_0)$ is an initial configuration of
  $\mA$, $T_j$ is labelled by $w_j$ for each $0\le j\le i$,
  $\range(T_j)=\dom(T_{j+1})$ for each $0\le j<i$, and $\range(T_i)=c$.  We
  show that this sequence is in fact a prefix of an accepting run of $\mA$
  over $w$ bounded by $Z$. 

  We inductively define a multitransition sequence $T_{i+1}T_{i+2}\ldots$
  completing this run.  The definition uses the suffix
  $r_{i+1}r_{i+2}\ldots$ of $\sigma(w)$.  Let us assume that $j>i$ and that
  $\range(T_{j-1})$ is a configuration of $\AC_Z$.  We define $T_j$ to
  contain one $w_j$-transition of $s$ for each $s\in\range(T_{j-1})$.  Thus
  we get $\dom(T_j)=\range(T_{j-1})$.  As $r_j\in\AT{Z}$, there exists a
  multitransition $T'$ labelled by $w_j$ such that both source and target
  configurations of $T'$ are in $\AC_Z$.  For each must-state
  $s\in\range(T_{j-1})$, $T_j$ contains the same transition leading from $s$
  as contained in $T'$.  For may-states $f\in\range(T_{j-1})$, we have two
  cases.  If $r_j\in\AT{Z}^f$, $T_j$ contains a non-looping transition
  leading from $f$ to some states in $Z$.  The existence of such a
  transition follows from the definition of $\AT{Z}^f$.  For the remaining
  may-states, $T_j$ uses selfloops.
  Formally, $T_j = \{ t_j^{s} \mid s\in\range(T_{j-1})\}$, where
$$
t_{j}^{s} = 
\begin{cases}
%(s,w_j,c_s) \text{, where } (s,w_j,c_s) \in T_j' & \text{if } s \in \must(Z) \\
(s,w_j,c_s) \text{ contained in }  T' & \text{if } s \in \must(Z) \\
(s,w_j,\{s\}) & \text{if } s \in F \land r_j \notin \AT{Z}^s \\
(s,w_j,c_s) \text{ where } c_s\subseteq Z, s \notin c_s & \text{if } 
s \in F \land r_j \in \AT{Z}^s \\
\end{cases}
$$
One can easily check that $range(T_j)\in\AC_Z$ and we continue by building
$T_{j+1}$.  

To sum up, the constructed run is bounded by $Z$.  Moreover, $T_j$ contains
no looping transition of $f$ whenever $r_j\in\AT{Z}^f$. As the run
$\sigma(w)$ is accepting, $r_j\in\AT{Z}^f$ holds infinitely often for each
$f\in F\cap Z$. The constructed run of $\mA$ over $w$ is thus accepting.\qed
\end{proof}

The previous two lemmata give us the following theorem.
\begin{theorem}
The TGDRA $\mG = (\mT,\{ \mGR_Z \mid Z \subseteq S\})$ describes the same
language as $\mA$.
\end{theorem}

\section{Translation of TGDRA to DRA}
This section presents a variant of the standard degeneralization procedure.  
At first we illustrate the idea on a TGDRA
$\mG'=(M,\Sigma,\delta_\mT,m_I,\{(K,\{L^j\}_{1\leq j \leq h})\})$ with one
generalized Rabin pair. Recall
that a run is accepting if it has a suffix not using macrotransitions of $K$
and using macrotransitions of each $L^j$ infinitely often.

An equivalent DRA $\mD'$ consists of $h+2$ copies of $\mG'$. The copies are
called \emph{levels}.  We start at the level $1$.  Intuitively, being at a
level $j$ for $1\le j\le h$ means that we are waiting for a transition from
$L^j$.  Whenever a transition of $K$ appears, we move to the level $0$.  A
transition $r\not\in K$ gets us from a level $j$ to the maximal level $l\ge
j$ such that $r\in L^{j'}$ for each $j\le j'< l$. The levels $0$ and $h+1$
have the same transitions (including target levels) as the level $1$. A run
of $\mG'$ is accepting if and only if the corresponding run of $\mD'$ visits
the level $0$ only finitely often and it visits the level $h+1$ infinitely
often.

% We call the copies levels.  The states of copies $0$ and $h+1$ form the
% sets $K'$ and $L'$, respectively.  Being in level $j$, $\mD'$ waits for a
% transition from $L^j$.  If $r$ from $K$ appears, the level is reset to
% $0$.  When $r \in L^j$ is used, $\mD'$ jumps to the maximal level $j'$
% such that $r \in L^l$ for all $j \leq l < j'$.  When the level $h+1$ is
% reached, it means that transitions from all sets $L^j$ were taken without
% taking any transition from $K$.  All transitions $r\notin K$ of levels $0$
% and $h+1$ lead to the maximal level $j'$ such that $r \in L^l$ for all $1
% \leq l < j'$.

%Let $\mG=(M,\Sigma,\delta_\mT,m_I,\{(K_i,\{L_i^j\}_{1\leq j\leq h_i})\}_{1\leq i 
%\leq k})$ be a TGDRA. We use the standard trick of degeneralization to transform
%$\mG$ into DRA $\mD$.  For fixed $1\leq i\leq k$, we
%have $h_i+2$ copies of $\mG$, we call the copies levels.  Being in level $j$,
%$\mD$ waits for a transition from $L_i^j$.  When $r \in L_i^j$ is used, it jumps
%to level $j+1$.  In fact, we can skip all levels $j < j'$ such that 
%$r\in L_i^{j'}$.
%After $\mD$ \kickoff{uses transitions from all requested sets it} reaches the 
%highest level for $i$ \kickoff{and} then resets to level 1.  In addition, whenever
%a transition from $K_i$ appears, it resets the level to 0 and then $\mD$ continues 
%from 1 again.

In the general case, we track the levels for all generalized Rabin pairs
simultaneously. Given a TGDRA
$\mG=(M,\Sigma,\delta_\mT,m_I,\{(K_i,\{L_i^j\}_{1\leq j\leq h_i})\}_{1\leq i
  \leq k})$, we construct an equivalent DRA as $\mD=(Q,\Sigma,
\delta_{\mD},q_i,\{(K'_i,L'_i)\}_{1\leq i \leq k})$, where
\begin{itemize}
\item $Q=M \times \{0,1,\ldots,h_1{+}1\} \times \cdots 
\times \{0,1,\ldots,h_k{+}1\}$,
\item $((m,l_1,\ldots,l_k),\alpha,(m',l'_1,\ldots,l'_k)) \in \delta_\mD$ 
iff $r=(m,\alpha,m') \in \delta_\mT$ and for each $1\leq i \leq k$ it holds
$$l'_i = \begin{cases}
0 &\text{if }r\in K_i \\
\max\{l_i \leq l \leq h_i{+}1 \mid \forall l_i \leq j < l:r\in L_i^j\} &\text{if }r \notin K_i \land 1 \leq l_i \leq h_i \\
\max\{1 \leq l \leq h_i{+}1 \mid \forall 1 \leq j < l:r\in L_i^j\} &\text{if }r\notin K_i\land l_i \in \{0,h_i{+}1\}\text{,} \\
\end{cases}$$
\item $q_i = (m_I,1,\ldots,1)$,
\item $K'_i= \{(m,l_1,\ldots,l_k)\in Q \mid l_i=0\}$, and
\item $L'_i= \{(m,l_1,\ldots,l_k)\in Q \mid l_i=h_i+1\}$.
\end{itemize}

%}}}
%{{{ Complexity

\section{Complexity}

This section discusses the upper bounds of the individual steps of our
translation and compares the overall complexity to the complexity of the
other translations.

Given a formula $\varphi$ of LTL($\Fs,\Gs$), we produce an MMAA with
at most $n$ states, where $n$ is the length of $\varphi$.  Then we build the
TGDRA $\mG$ with at most $2^{2^{n}}$ states and at most $2^n$ generalized
Rabin pairs.  To obtain the DRA $\mD$, we multiply the state space by at
most $|Z|+2$ for each generalized Rabin pair $\mGR_Z$.  The value of $|Z|$
is bounded by $n$. Altogether, we can derive an upper bound on the number of
states of the resulting DRA as
$$
|Q| \leq 
%2^{2^{n}}\cdot\!\!\!\!\!\!\!\prod_{i \in \{1,\ldots,2^n\}}\!\!\!\!\!\!\!\!(h_i{+}2)\leq
2^{2^{n}} \cdot (n+2)^{2^{n}} =
2^{2^{n}} \cdot 2^{{2^n}\cdot\log_2{(n+2)}} = 
2^{2^{n}} \cdot 2^{2^{n+\log_2\log_2(n+2)}} \in
2^{\mO(2^{n+\log \log n})},
$$
which is the same bound as in \cite{KE12}, but lower than $2^{\mO(2^{n+\log
    n})}$ of \texttt{ltl2dstar}.  It is worth mentioning that the number of
states of our TGDRA is bounded by $2^{2^{|\varphi|}}$ while the number of
states of the GDRA produced by Rabinizer is bounded by
$2^{2^{|\varphi|}}\cdot 2^{\AP(\varphi)}$. 

%}}}
%{{{ Simplification

\section{Simplifications and Translation Improvements}
\label{sec:simplifications}

An important aspect of our translation process is simplification of all
intermediate results leading to smaller resulting DRA.

We simplify input formulae by reduction rules of LTL3BA,
see~\cite{BKRS12} for more details. Additionally, we rewrite the
subformulae of the form $\G\F\psi$ and $\F\G\psi$ to equivalent
formulae $\G\Fs\psi$ and $\F\Gs\psi$ respectively. This preference of
strict temporal operators often yields smaller resulting automata.

Alternating automata are simplified in the same way as in LTL2BA: removing
unreachable states, merging equivalent states, and removing redundant
transitions, see~\cite{GO01} for details.

We improve the translation of an MMAA $\mA$ to a TGDRA $\mG$ in order to
reduce the number of generalized Rabin pairs of $\mG$.  One can observe
that, for any accepting run $\rho$ of $\mA$, $\Inf_s(\rho)$ contains
only states reachable from some must-state.  Hence, in the construction of
acceptance condition of $\mG$ we can consider only subsets $Z$ of states
of $\mA$ of this form.  Further, we omit a subset $Z$ if, for each
accepting run over $w$ bounded by $Z$, there is also an accepting run
over $w$ bounded by some $Z'\subseteq Z$.  The formal description of
subsets $Z$ considered in the construction of the TGDRA $\mG$ is described
in Appendix~\ref{app:Zset}.

If a run $T_0T_1\ldots$ of an MMAA satisfies $\range(T_i)=\emptyset$ for
some $i$, then $T_j=\emptyset$ for all $j \geq i$ and the run is accepting.
We use this observation to improve the construction of the semiautomaton
$\mT$ of the TGDRA $\mG$: if a macrostate $m$ contains the empty
configuration, we remove all other configurations from $m$.
% The semiautomaton $\mT$ then contains a unique macrostate $\{\emptyset\}$.

After we build the TGDRA, we simplify its acceptance condition in three
ways (similar optimizations are also performed by Rabinizer).
\begin{enumerate}
\item We remove some generalized Rabin pairs $(K_i,\{L_{i}^{j}\}_{j\in J_i})$
  that cannot be satisfied by any run, in particular when $K_i =
  \delta_{\mT}$ or $L_i^j=\emptyset$ for some $j \in J_i$.
\item We remove $L_{i}^{j}$ if there is some $l \in J_i$ such that $L_i^{l} 
\subseteq L_i^j$.
\item If the fact that a run $\rho$ satisfies the pair $\mGR_Z$ implies that
  $\rho$ satisfies also some other pair $\mGR_{Z'}$, we remove $\mGR_Z$.
\end{enumerate}

Finally, we simplify the state spaces of both TGDRA and DRA such that we
iteratively merge the equivalent states. Two states of a DRA $\mD$ are
equivalent if they belong to the same sets of the acceptance condition of
$\mD$ and, for each $\alpha$, their $\alpha$-transitions lead to the same
state. Two states of a TGDRA $\mG$ are equivalent if, for each $\alpha$,
their $\alpha$-transitions lead to the same state and belong to the same
sets of the acceptance condition of $\mG$. Moreover, if the initial state of
$\mD$ or $\mG$ has no selfloop, we check its equivalence to another state
regardless of the acceptance condition (note that a membership in acceptance
condition sets is irrelevant for states or transitions that are passed at
most once by any run).

% The following table shows when two states $s_1$ and $s_2$ are equivalent in
% TGDRA and DRA.
%\kickoff{
%\begin{center}\todo{dodat kvantifikatory? Bude to hnus. Zafixovat automaty?}
%\begin{tabular}{|c|c|}
%\hline 
%TGDRA & $t_1=(s_1,\alpha,s) \in \delta \Leftrightarrow t_2=(s_2,\alpha,s) \in \delta$ \\ & $t_1\in K_i \Leftrightarrow  t_2\in K_i$, and $t_1\in L_i^j \Leftrightarrow  t_2\in L_i^j$ \\ 
%\hline 
%DRA & $(s_1,\alpha,s)\in \delta \Leftrightarrow
%(s_2,\alpha,s)\in \delta$ \\
% & $s_1\in K_i \Leftrightarrow s_2 \in K_i \land s_1\in L_i \Leftrightarrow s_2 \in L_i$ \\
%\hline 
%\end{tabular}
%\end{center}}

Of course, we consider only the reachable state space at every step.

%}}}
%{{{ May/must in the limit 

\section{Beyond LTL($\F_s$,$\G_s$) Fragment: May/Must in the Limit}
\label{sec:limmmaa}

The Section \ref{sec:MMAA2DRA} shows a translation of MMAA into TGDRA.  In
fact, our translation can be used for a larger class of very weak
alternating automata called \emph{may/must in the limit automata} (limMMAA).
A VWAA $\mB$ is a limMMAA if $\mB$ contains only must-states, states without
looping transitions, and co-B\"{u}chi accepting states (not exclusively
may-states), and each state reachable from a must-state is either a must- 
or a may-state. Note that each accepting run of a limMMAA has a suffix that
contains either only empty configurations, or configurations consisting of
must-states and may-states reachable from must-states. Hence, the MMAA to
TGDRA translation produces correct results also for limMMAA under an
additional
condition: generalized Rabin pairs $\mGR_Z$ are constructed only for sets
$Z$ that contain only must-states and may-states reachable from them.

We can obtain limMMAA by the LTL to VWAA translation of \cite{GO01} when it
is applied to an LTL fragment defined as
% $\varphi\,::=\,\psi\mid\varphi\vee\varphi\mid\varphi\wedge\varphi\mid
% \X\varphi\mid\varphi\U\varphi$,
\[\varphi~::=~\psi~\mid~\varphi\vee\varphi~\mid~\varphi\wedge\varphi~\mid~
\X\varphi~\mid~\varphi\U\varphi\text{,}
\]
where $\psi$ ranges over LTL($\Fs$,$\Gs$). Note that this fragment is
strictly more expressive than LTL($\Fs$,$\Gs$).

%\begin{theorem}
%Let $\mB$ be an limMMAA. The TGDRA $\mG$ constructed as described in Section
%\ref{sec:MMAA2DRA} is equivalent to $\mB$.
%\end{theorem}
% \todo{dat jako theorem?}
% Let $\rho$ be an accepting run of some limMMAA $\mB$ over $w$. $\Inf_s(\rho)$ 
% contains only must- and may-states.  Otherwise $\rho$ would loop in some co-
% B\"{u}chi state and thus would not be accepting.  Now the Lemma \ref{lem:ag}
% applies.  For the other direction we can apply the Lemma \ref{lem:ga} directly.

%}}}
%{{{ Experimental results

\section{Experimental Results}
\label{sec:experiments}
We have made an experimental implementation of our translation (referred to
as \emph{LTL3DRA}). The translation of LTL to alternating automata is taken
from LTL3BA~\cite{BKRS12}.  We compare the automata produced
by LTL3DRA to those produced by Rabinizer and \texttt{ltl2dstar}. All the
experiments are run on a Linux laptop (2.4GHz Intel Core i7, 8GB of RAM)
with a timeout set to 5 minutes.

Tables given below (i) compare the sizes of the DRA produced by all the
tools and (ii) show the number of states of the generalized automata
produced by LTL3DRA and Rabinizer. Note that LTL3DRA uses TGDRA
whereas Rabinizer uses (state-based) GDRA,
hence the numbers of their states cannot be
directly compared. The sizes of DRA are written as $s(r)$, where $s$ is the
number of states and $r$ is the number of Rabin pairs.  For each formula,
the size of the smallest DRA (measured by the number of states and, in the
case of equality, by the number of Rabin pairs) is printed in bold.

%{{{ Table 1

\begin{table}[t!]
\centering
\setlength{\tabcolsep}{3pt}
\small
\begin{tabular}{|c|c||rc|rr|r|}
\hline
\multicolumn{2}{ |c|| }{\multirow{2}{*}{Formula} } &
  \multicolumn{2}{ |c| }{LTL3DRA} & 
  \multicolumn{2}{ c| }{Rabinizer} &
  \multicolumn{1}{ |c| }{\texttt{ltl2dstar}}
  \\ \cline{3-7}
\multicolumn{2}{ |c|| }{} &
\multicolumn{2}{|r|}{\scriptsize DRA~~~TGDRA} & \multicolumn{2}{|r|}{\scriptsize DRA~~GDRA} & \multicolumn{1}{ |c|  }{\scriptsize DRA}   \\ \cline{1-7}
%\hline%\hline
\multicolumn{2}{ |c|| }{$ \G(a \lor \F b)$}& \textbf{3(2)}& 2 & 4(2) & 5~ & 4(1) \\
%\hline
\multicolumn{2}{ |c|| }{$\F\G a \lor \F\G b \lor \G\F c$}& \textbf{8(3)}& 1 & \textbf{8(3)} & 8~ & \textbf{8(3)} \\
%\hline
\multicolumn{2}{ |c|| }{$ \F(a \lor b)$}& \textbf{2(1)}& 2 & \textbf{2(1)} & 2~ & \textbf{2(1)} \\
%\hline
\multicolumn{2}{ |c|| }{$ \G\F(a \lor b)$}& \textbf{2(1)}& 1 & \textbf{2(1)} & 4~ & \textbf{2(1)} \\
%\hline
\multicolumn{2}{ |c|| }{$ \G(a \lor \F a)$}& \textbf{2(1)}& 1 & 2(2) & 2~ & \textbf{2(1)} \\
%\hline
\multicolumn{2}{ |c|| }{$ \G(a \lor b \lor c)$}& \textbf{2(1)}& 2 & \textbf{2(1)} & 8~ & 3(1) \\
%\hline
\multicolumn{2}{ |c|| }{$ \G(a \lor \F(b \lor c))$}& \textbf{3(2)}& 2 & 4(2) & 9~ & 4(1) \\
%\hline
\multicolumn{2}{ |c|| }{$ \F a \lor \G b$}& \textbf{3(2)}& 3 & \textbf{3(2)} & 3~ & 4(2) \\
%\hline
\multicolumn{2}{ |c|| }{$ \G(a \lor \F(b \land c))$}& \textbf{3(2)}& 2 & 4(2) & 11~ & 4(1) \\
%\hline
\multicolumn{2}{ |c|| }{$ \F\G a \lor \G\F b$}& \textbf{4(2)}& 1 & \textbf{4(2)} & 4~ & \textbf{4(2)} \\
%\hline
\multicolumn{2}{ |c|| }{$ \G\F(a \lor b) \land \G\F(b \lor c)$}& \textbf{3(1)}& 1 & \textbf{3(1)} & 8~ & 7(2) \\
%\hline
\multicolumn{2}{ |c|| }{$ (\F\F a \land \G \neg a) \lor (\G\G \neg a \land \F a)$}& \textbf{1(0)}& 1 & \textbf{1(0)} & 1~ & \textbf{1(0)} \\
%\hline
\multicolumn{2}{ |c|| }{$ \G\F a \land \F\G b$}& \textbf{3(1)}& 1 & \textbf{3(1)} & 4~ & \textbf{3(1)} \\
%\hline
\multicolumn{2}{ |c|| }{$ ( \G\F a \land \F\G b ) \lor (\F\G\neg a \land \G \F \neg b ) $}& \textbf{4(2)}& 1 & \textbf{4(2)} & 4~ & 5(2) \\
%\hline
\multicolumn{2}{ |c|| }{$ \F\G a \land \G\F a$}& \textbf{2(1)}& 1 & \textbf{2(1)} & 2~ & \textbf{2(1)} \\
%\hline
\multicolumn{2}{ |c|| }{$ \G ( \F a \land \F b)$}& \textbf{3(1)}& 1 & \textbf{3(1)} & 4~ & 5(1) \\
%\hline
\multicolumn{2}{ |c|| }{$ \F a \land \F \neg a$}& \textbf{4(1)}& 4 & \textbf{4(1)} & 4~ & \textbf{4(1)} \\
%\hline
\multicolumn{2}{ |c|| }{$ ( \G ( b \lor \G\F a ) \land \G ( c \lor \G\F \neg a)) \lor \G b \lor \G c$}& \textbf{12(3)}& 4 & 18(4) & 18~ & 13(3) \\
%\hline
\multicolumn{2}{ |c|| }{$ ( \G ( b \lor \F\G a ) \land \G ( c \lor \F\G \neg a)) \lor \G b \lor \G c$}& \textbf{4(2)}& 4 & 6(3) & 18~ & 14(4) \\
%\hline
\multicolumn{2}{ |c|| }{$ ( \F ( b \land \F\G a ) \lor \F ( c \land \F\G \neg a)) \land \F b \land \F c$}& \textbf{5(2)}& 4 & \textbf{5(2)} & 18~ & 7(1) \\
%\hline
\multicolumn{2}{ |c|| }{$ ( \F ( b \land \G\F a ) \lor \F ( c \land \G\F \neg a)) \land \F b \land \F c$}& \textbf{5(2)}& 4 & \textbf{5(2)} & 18~ & 7(2) \\
%\hline
\multicolumn{2}{ |c|| }{$ \G\F(\F a \lor \G\F b \lor \F\G(a \lor b ))$}& \textbf{4(3)}& 1 & \textbf{4(3)} & 4~ & 14(4) \\
%\hline
\multicolumn{2}{ |c|| }{$ \F\G(\F a \lor \G\F b \lor \F\G(a \lor b ))$}& \textbf{4(3)}& 1 & \textbf{4(3)} & 4~ & 145(9) \\
%\hline
\multicolumn{2}{ |c|| }{$ \F\G(\F a \lor \G\F b \lor \F\G(a \lor b ) \lor \F\G b)$}& \textbf{4(3)}& 1 & \textbf{4(3)} & 4~ & 145(9) \\
\hline
%\hline
\multirow{4}{*}{$\bigwedge_{i=1}^n(\G\F a_{i} \rightarrow \G\F b_{i})$}
&$n=1$&\textbf{4(2)}& 1 & \textbf{4(2)} & 4~ & \textbf{4(2)} \\
%\cline{2-7}
&$n=2$&\textbf{18(4)}& 1 & 20(4) & 16~ & 11324(8) \\
%\cline{2-7}
&$n=3$&\textbf{166(8)}& 1 & 470(8) & 64~ & \multicolumn{1}{|c|}{timeout}\\
%\cline{2-7}
&$n=4$&\textbf{7408(16)}& 1 & \multicolumn{2}{|c|}{timeout} & \multicolumn{1}{|c|}{timeout}\\
\hline
\multirow{5}{*}{$\bigwedge_{i=1}^n(\G\F a_i \lor \F\G a_{i+1})$}
&$n=1$&\textbf{4(2)}& 1 & \textbf{4(2)} & 4~ & \textbf{4(2)} \\
%\cline{2-7}
&$n=2$&\textbf{10(4)}& 1 & 11(4) & 8~ & 572(7) \\
%\cline{2-7}
&$n=3$&\textbf{36(6)}& 1 & 52(6) & 16~ & 290046(13) \\
%\cline{2-7}
&$n=4$&\textbf{178(9)}& 1 & 1288(9) & 32~ & \multicolumn{1}{|c|}{timeout}\\
%\cline{2-7}
&$n=5$&\textbf{1430(14)}& 1 & \multicolumn{2}{|c|}{timeout} & \multicolumn{1}{|c|}{timeout}\\
%\cline{2-7}
&$n=6$&\textbf{20337(22)}& 1 & \multicolumn{2}{|c|}{timeout} & \multicolumn{1}{|c|}{timeout}\\
\hline
\multicolumn{3}{c}{}&\multicolumn{1}{c}{~~~~~~~~~}&\multicolumn{1}{c}{}&\multicolumn{1}{c}{~~~~~~}
\end{tabular}
\smallskip
\caption{The benchmark from \cite{GKE12} extended by one parametric formula.} 
\label{tab:ATVA}
\end{table}

%}}}

Table \ref{tab:ATVA} shows the results on formulae from \cite{GKE12}
extended with another parametric formula. For the two parametric formulae,
we give all the parameter values $n$ for which at least one tool finished
before timeout. For all formulae in the table, our experimental
implementation generates automata of the same or smaller size as the others.
Especially in the case of parametric formulae, the automata produced by
LTL3DRA are considerably smaller. We also note that the TGDRA constructed
for the formulae are typically very small.
% Further, the results on parametric formulae indicates that our
% implementation quite often scales significantly better.

Table~\ref{tab:SPEC} shows the results on formulae from \textsc{Spec
  Patterns}~\cite{DAC99} (available
online\footnote{\url{http://patterns.projects.cis.ksu.edu/documentation/patterns/ltl.shtml}}).
We only take formulae LTL3DRA is able to work with, i.e.~the formulae of the
LTL fragment defined in Section~\ref{sec:limmmaa}.  The fragment covers 27
out of 55 formulae listed on the web page.  The dash sign in Rabinizer's
column means that Rabinizer cannot handle the corresponding formula as it is
not from the LTL($\F,\G$) fragment.  For most of the formulae in the table,
LTL3DRA produces the smallest DRA. In the remaining cases, the DRA produced
by our translation is only slightly bigger than the smallest one.  The table
also illustrates that LTL3DRA handles many (pseudo)realistic formulae not
included in LTL($\F,\G$).  
Four more parametric benchmarks are provided in Appendix~\ref{sec:moreexp}.
% provides experimental results for four more parametric formulae.

%{{{ Table 2

\begin{table}[t!]
\centering
\begin{tabular}{|c||r c|rc|r|c|c||rc|rc|r|}
%\hline
\cline{1-6}\cline{8-13}
\multicolumn{1}{ |c|| }{\multirow{2}{*}{} } &
  \multicolumn{2}{ |c| }{LTL3DRA} & 
  \multicolumn{2}{ c| }{Rabinizer} &
  \multicolumn{1}{ |c| }{\texttt{ltl2dstar}}&
 \multicolumn{1}{c}{~~} &
\multicolumn{1}{ |c|| }{\multirow{2}{*}{} } &
  \multicolumn{2}{ |c| }{LTL3DRA} & 
  \multicolumn{2}{ c| }{Rabinizer} &
  \multicolumn{1}{ |c| }{\texttt{ltl2dstar}}
 \\  \cline{2-6}\cline{9-13}
\multicolumn{1}{ |c||  }{} &
\multicolumn{2}{|c|}{\scriptsize DRA~~TGDRA} & \multicolumn{2}{|c|}{\scriptsize DRA~~GDRA} & \multicolumn{1}{ |c|  }{\scriptsize DRA} &
\multicolumn{1}{ |c|  }{} & \multicolumn{1}{ |c||  }{} &
\multicolumn{2}{|c|}{\scriptsize DRA~~TGDRA} & \multicolumn{2}{|c|}{\scriptsize DRA~~GDRA} & \multicolumn{1}{ |c|  }{\scriptsize DRA}\\  
\cline{1-6}\cline{8-13}
%\cline{1-6}%\cline{8-13}
$\varphi_{2}$ & \textbf{4(2)}& 4 & \multicolumn{2}{|c|}{---}  & 5(2)~~~ & & $\varphi_{27}$ & \textbf{4(2)}& 4 & \multicolumn{2}{|c|}{---}  & 5(2)~~~~\\  
%\hline
$\varphi_{3}$ & ~~4(2)& 3 & ~~4(2) & 5 & \textbf{4(1)}~~~ & & $\varphi_{28}$ & ~~6(3)& 3 & ~~8(3) & 14 & \textbf{5(1)}~~~~\\ 
%\cline{1-6}%\cline{8-13}
$\varphi_{7}$ & \textbf{4(2)}& 3 & \multicolumn{2}{|c|}{---}  & \textbf{4(2)} ~~~&& $\varphi_{31}$ & \textbf{4(2)}& 4 & \multicolumn{2}{|c|}{---}  & 6(2)~~~~\\  
%\cline{1-6}%\cline{8-13}
$\varphi_{8}$ & \textbf{3(2)}& 3 & \textbf{3(2)} & 5 & 4(2) ~~~&& $\varphi_{32}$ & \textbf{5(2)}& 5 & \multicolumn{2}{|c|}{---}  & 7(2)~~~~\\  
%\cline{1-6}%\cline{8-13}
$\varphi_{11}$ & \textbf{6(2)}& 6 & \multicolumn{2}{|c|}{---}  & 10(3) ~~~&& $\varphi_{33}$ & \textbf{5(2)}& 5 & \multicolumn{2}{|c|}{---}  & 7(3)~~~~\\ 
%\cline{1-6}%\cline{8-13}
$\varphi_{12}$ & \textbf{8(2)}& 8 & \multicolumn{2}{|c|}{---}  & 9(2) ~~~&&  $\varphi_{36}$ & 6(3)& 4 & \multicolumn{2}{|c|}{---}  & \textbf{6(2)}~~~~\\ 
%\cline{1-6}%\cline{8-13}
$\varphi_{13}$ & \textbf{7(3)}& 7 & \multicolumn{2}{|c|}{---}  & 11(3) ~~~&&  $\varphi_{37}$ & \textbf{6(2)}& 6 & \multicolumn{2}{|c|}{---}  & 8(3)~~~~\\ 
%\cline{1-6}%\cline{8-13}
$\varphi_{17}$ & \textbf{4(2)}& 4 & \multicolumn{2}{|c|}{---}  & 5(2) ~~~&&  $\varphi_{38}$ & 7(4)& 5 & \multicolumn{2}{|c|}{---}  & \textbf{6(3)}~~~~\\  
%\cline{1-6}%\cline{8-13}
$\varphi_{18}$ & 4(2)& 3 & 4(2) & 5 & \textbf{4(1)} ~~~&&  $\varphi_{41}$ & \textbf{21(3)}& 7 & \multicolumn{2}{|c|}{---}  & 45(3)~~~~\\  
%\cline{1-6}%\cline{8-13}
$\varphi_{21}$ & \textbf{4(2)}& 3 & \multicolumn{2}{|c|}{---}  & \textbf{4(2)} ~~~&&  $\varphi_{42}$ & \textbf{12(2)}& 12~ & \multicolumn{2}{|c|}{---}  & 17(2)~~~~\\  
%\cline{1-6}%\cline{8-13}
$\varphi_{22}$ & \textbf{4(2)}& 4 & \multicolumn{2}{|c|}{---}  & 5(2) ~~~&&  $\varphi_{46}$ & \textbf{15(3)}& 5 & \multicolumn{2}{|c|}{---}  & 20(2)~~~~\\  
%\cline{1-6}%\cline{8-13}
$\varphi_{23}$ & \textbf{5(3)}& 4 & \multicolumn{2}{|c|}{---}  & \textbf{5(3)} ~~~&&  $\varphi_{47}$ & 7(2)& 7 & \multicolumn{2}{|c|}{---} & \textbf{6(2)}~~~~\\  
%\cline{1-6}%\cline{8-13}
$\varphi_{26}$ & \textbf{3(2)}& 2 & 4(2) & 5 & 4(1) ~~~&&  $\varphi_{48}$ & \textbf{14(3)}& 6 & \multicolumn{2}{|c|}{---}  & 24(2)~~~~\\  
\cline{1-6}
%\cline{8-13}
\multicolumn{7}{c|}{}&$\varphi_{52}$ & 7(2)& 7 & \multicolumn{2}{|c|}{---}  & \textbf{6(2)}~~~~\\  
\cline{8-13}
\end{tabular}
\caption{The benchmark with selected formulae from \textsc{Spec
    Patterns}. $\varphi_i$ denotes the $i$-th formula on the web page.}
\label{tab:SPEC}
\end{table}

%}}}

%}}}
%{{{ Conclusion

\section{Conclusion}\label{sec:concl}
We present another Safraless translation of an LTL fragment to deterministic
Rabin automata (DRA).  Our translation employs a new class of \emph{may/must
  alternating automata}.  We prove that the class is expressively equivalent
to the LTL($\Fs,\Gs$) fragment.  Experimental results show that our
translation typically produces DRA of a smaller or equal size as the other
two translators of LTL (i.e.~Rabinizer and \texttt{ltl2dstar}) and it
sometimes produces automata that are significantly smaller.

% In future work, we would like to extend our approach to the whole LTL (or at
% least to a bigger fragment). We also plan to introduce more simplifications
% of the translation and to improve performance of our experimental
% implementation. Another interesting topic is to modify current algorithms
% working with DRA to work directly with TGDRA. As TGDRA are often
% significantly smaller than DRA, we expect the modified algorithms to run faster.

%}}}

\bibliographystyle{abbrv}
\bibliography{atva}

\begin{thebibliography}{10}

\bibitem{AT04}
R.~Alur and S.~L. Torre.
\newblock Deterministic generators and games for {LTL} fragments.
\newblock {\em ACM Trans. Comput. Log.}, 5(1):1--25, 2004.

\bibitem{BBDL13}
T.~Babiak, T.~Badie, A.~Duret-Lutz, M.~K\v{r}et\'{\i}nsk{\'y}, and
  J.~Strej\v{c}ek.
\newblock Compositional approach to suspension and other improvements to {LTL}
  translation.
\newblock In {\em {SPIN} 2013}, volume 7976 of {\em LNCS}, pages 81--98.
  Springer, 2013.

\bibitem{BKRS12}
T.~Babiak, M.~K{\v{r}}et\'{\i}nsk{\'y}, V.~{\v{R}}eh{\'a}k, and
  J.~Strej{\v{c}}ek.
\newblock {LTL} to {B}{\"u}chi automata translation: Fast and more
  deterministic.
\newblock In {\em {TACAS} 2012}, volume 7214 of {\em LNCS}, pages 95--109.
  Springer, 2012.

\bibitem{BK08}
C.~Baier and J.-P. Katoen.
\newblock {\em Principles of Model Checking}.
\newblock MIT Press, 2008.

\bibitem{CGK13}
K.~Chatterjee, A.~Gaiser, and J.~K\v{r}et\'{\i}nsk\'{y}.
\newblock Automata with generalized {R}abin pairs for probabilistic model
  checking and {LTL} synthesis.
\newblock In {\em {CAV} 2013}, volume 8044 of {\em LNCS}, pages 559--575.
  Springer, 2013.

\bibitem{Church62}
A.~Church.
\newblock Logic, arithmetic, and automata.
\newblock In {\em Proceedings of the International Congress of Mathematicians},
  pages 23--35. Institut Mittag-Leffler, 1962.

\bibitem{CGP99}
E.~M. Clarke, O.~Grumberg, and D.~A. Peled.
\newblock {\em {M}odel {C}hecking}.
\newblock MIT Press, 1999.

\bibitem{CY95}
C.~Courcoubetis and M.~Yannakakis.
\newblock The complexity of probabilistic verification.
\newblock {\em J. ACM}, 42(4):857--907, 1995.

\bibitem{Cou99}
J.-M. Couvreur.
\newblock On-the-fly verification of temporal logic.
\newblock In {\em {FM} 1999}, volume 1708 of {\em LNCS}, pages 253--271.
  Springer, 1999.

\bibitem{DGV99}
M.~Daniele, F.~Giunchiglia, and M.~Y. Vardi.
\newblock Improved automata generation for linear temporal logic.
\newblock In {\em {CAV} 1999}, volume 1633 of {\em LNCS}, pages 249--260.
  Springer, 1999.

\bibitem{DL11}
A.~Duret-Lutz.
\newblock {LTL} translation improvements in {Spot}.
\newblock In {\em {VECoS} 2011}, Electronic Workshops in Computing. British
  Computer Society, 2011.

\bibitem{DAC99}
M.~B. Dwyer, G.~S. Avrunin, and J.~C. Corbett.
\newblock Patterns in property specifications for finite-state verification.
\newblock In {\em ICSE 1999}, pages 411--420. IEEE, 1999.

\bibitem{GKE12}
A.~Gaiser, J.~K{\v{r}}et{\'\i}nsk{\'y}, and J.~Esparza.
\newblock Rabinizer: Small deterministic automata for {LTL(F,G)}.
\newblock In {\em {ATVA} 2012}, volume 7561 of {\em LNCS}, pages 72--76, 2012.

\bibitem{GO01}
P.~Gastin and D.~Oddoux.
\newblock Fast {LTL} to {B}{\"u}chi {A}utomata {T}ranslation.
\newblock In {\em {CAV} 2001}, volume 2102 of {\em LNCS}, pages 53--65.
  Springer, 2001.

\bibitem{GH06}
J.~Geldenhuys and H.~Hansen.
\newblock Larger automata and less work for {LTL} model checking.
\newblock In {\em {SPIN} 2006}, volume 3925 of {\em LNCS}, pages 53--70.
  Springer, 2006.

\bibitem{Kle}
J.~Klein.
\newblock ltl2dstar -- {LTL} to deterministic {S}treett and {R}abin automata.
\newblock http://www.ltl2dstar.de.

\bibitem{KB06}
J.~Klein and C.~Baier.
\newblock Experiments with deterministic {$\omega$}-automata for formulas of
  linear temporal logic.
\newblock {\em Theor. Comput. Sci.}, 363(2):182--195, 2006.

\bibitem{KB07}
J.~Klein and C.~Baier.
\newblock On-the-fly stuttering in the construction of deterministic
  {$\omega$}-automata.
\newblock In {\em CIAA 2007}, volume 4783 of {\em LNCS}, pages 51--61.
  Springer, 2007.

\bibitem{KE12}
J.~K{\v{r}}et{\'\i}nsk{\'y} and J.~Esparza.
\newblock Deterministic automata for the {(F, G)}-fragment of {LTL}.
\newblock In {\em {CAV} 2012}, volume 7358 of {\em LNCS}, pages 7--22.
  Springer, 2012.

\bibitem{Kupferman12}
O.~Kupferman.
\newblock Recent challenges and ideas in temporal synthesis.
\newblock In {\em SOFSEM 2012}, volume 7147 of {\em LNCS}, pages 88--98.
  Springer, 2012.

\bibitem{KNP11}
M.~Kwiatkowska, G.~Norman, and D.~Parker.
\newblock {PRISM} 4.0: Verification of probabilistic real-time systems.
\newblock In {\em {CAV} 2011}, volume 6806 of {\em LNCS}, pages 585--591.
  Springer, 2011.

\bibitem{MorgensternS08}
A.~Morgenstern and K.~Schneider.
\newblock From {LTL} to symbolically represented deterministic automata.
\newblock In {\em {VMCAI} 2008}, volume 4905 of {\em LNCS}, pages 279--293.
  Springer, 2008.

\bibitem{Pit07}
N.~Piterman.
\newblock From nondeterministic {B}{\"u}chi and {S}treett automata to
  deterministic parity automata.
\newblock {\em Logical Methods in Computer Science}, 3(3), 2007.

\bibitem{PPS06}
N.~Piterman, A.~Pnueli, and Y.~Sa'ar.
\newblock Synthesis of reactive(1) designs.
\newblock In {\em {VMCAI} 2006}, volume 3855 of {\em LNCS}, pages 364--380.
  Springer, 2006.

\bibitem{Pnu77}
A.~Pnueli.
\newblock The temporal logic of programs.
\newblock In {\em {FOCS} 1977}, pages 46--57. IEEE, 1977.

\bibitem{PnueliR89}
A.~Pnueli and R.~Rosner.
\newblock On the synthesis of an asynchronous reactive module.
\newblock In {\em {ICALP} 1989}, volume 372 of {\em LNCS}, pages 652--671.
  Springer, 1989.

\bibitem{Saf88}
S.~Safra.
\newblock On the complexity of omega-automata.
\newblock In {\em {FOCS} 1988}, pages 319--327. IEEE Computer Society, 1988.

\bibitem{Sch09}
S.~Schewe.
\newblock Tighter bounds for the determinisation of {B}{\"u}chi automata.
\newblock In {\em {FOSSACS} 2009}, volume 5504 of {\em LNCS}, pages 167--181.
  Springer, 2009.

\bibitem{SB00}
F.~Somenzi and R.~Bloem.
\newblock Efficient {B}\"{u}chi automata from {LTL} formulae.
\newblock In {\em {CAV} 2000}, volume 1855 of {\em LNCS}, pages 248--263.
  Springer, 2000.

\bibitem{Vardi85}
M.~Y. Vardi.
\newblock Automatic verification of probabilistic concurrent finite-state
  programs.
\newblock In {\em {FOCS} 1985}, pages 327--338. IEEE Computer Society, 1985.

\bibitem{VardiW86}
M.~Y. Vardi and P.~Wolper.
\newblock An automata-theoretic approach to automatic program verification.
\newblock In {\em {LICS} 1986}, pages 332--344. IEEE Computer Society, 1986.

\end{thebibliography}

\newpage
\appendix

%{{{ MMAA->LTL(Fs,Gs)

\section{Translation of MMAA to LTL($\Fs,\Gs$)}\label{sec:mmaa2ltl}
We can assume that may-states have no looping transitions except
selfloops. Indeed, any application of a looping transition that is not a
selfloop can be always replaced by an application of a selfloop with the
same label. This change is safe as it cannot transform an accepting run into
a non-accepting one.

Let $\mA=(Q,2^{\AP'},\delta,I,F)$ be an MMAA. For any $\alpha\in 2^{\AP'}$, we 
define $\psi_\alpha$ to be a formula satisfied exactly by the words starting 
with $\alpha$:
$$\psi_\alpha=
(\bigwedge_{a\in\alpha}a)\wedge(\bigwedge_{a\in\AP'\smallsetminus\alpha}\neg
a)$$ 
Now we inductively define a formula $\varphi_s$ for each $s\in Q$.  The
formula $\varphi_s$ is satisfied by any word for which there is an
accepting run of $\mA$ starting in the configuration $\{s\}$.  Admissibility of
the inductive definition follows from the fact that $\mA$ is a very weak
automaton, i.e.~there is a partial order on $Q$ such that transitions of a
state $s$ can lead only to $s$ or lower states.
\[
\varphi_s=\left\{
  \begin{array}{ll}%\displaystyle
    \F\bigvee_{(s,\alpha,c)\in\delta,
      c\neq\{s\}}(\psi_\alpha\wedge\bigwedge_{q\in c}\X\varphi_q) & 
    \textrm{ if $s$ is a may-state}\\
    \G\bigvee_{(s,\alpha,c)\in\delta}(\psi_\alpha\wedge\bigwedge_{q\in
      c\smallsetminus\{s\}}\X\varphi_q) &
    \textrm{ if $s$ is a must-state}\\      
    \bigvee_{(s,\alpha,c)\in\delta}(\psi_\alpha\wedge\bigwedge_{q\in
      c}\X\varphi_q) &
    \textrm{ if $s$ is a loopless state}\\      
  \end{array}
  \right.
\]
Note that the conjunction of an empty set of conjuncts is $\true$ while
the disjunction of an empty set of disjuncts is $\neg\true$. It is easy to see
that each temporal operator $\X$ in $\varphi_s$ is in front of $\F$ or $\G$.
If we replace all occurrences of $\X\F$ by $\Fs$ and all occurrences of
$\X\G$ by $\Gs$,
% in all formulae $\varphi_s$
we always get formulae of LTL($\Fs,\Gs$).

Finally, we define the formula $\varphi_\mA$ equivalent to the whole
automaton $\mA$ as
$$\varphi_\mA=\bigvee_{c\in I}\bigwedge_{s\in c}\varphi_s.$$

Hence, we have shown that the following theorem holds.
\begin{theorem}
  For each MMAA $\mA$ with an alphabet of the form $2^{\AP'}$, we can
  construct an LTL($\Fs,\Gs$) formula $\varphi_\mA$ such that
  $L(\mA)=L(\varphi_\mA)$.
\end{theorem}

%}}}
%{{{ Bounding Sets 

\section{Bounding Sets}\label{app:Zset}

Let $\mA = (S,\Sigma,\delta\A,I,F)$ be an MMAA or a limMMAA.  Here
we show how to compute the set $\mZ \subseteq 2^S$ of configurations
for which the generalized Rabin pairs $\mGR_Z$ of the
corresponding TGDRA are constructed.  If $\rho$ is an accepting run of $\mA$, then
$\Inf_s(\rho)$ is a subset of states reachable from must-states.  Further,
for each $w\in L(\mA)$ there is an accepting run $\rho$ over $w$ such that,
for each $f\in F$, either $f\notin\Inf_s(\rho)$, or $\rho$ uses a single
non-looping transition of $f$ (infinitely often) and the selfloop for~$f$.
We say that these runs are \emph{modest}.

The function $z:S\rightarrow 2^{2^S}$ recursively computes, for a given
state $s$, the sets of states that can potentially be reached from $s$ (not
necessary as one configuration) infinitely often by a modest run visiting
$s$ infinitely often. If $s$ is a may-state, a modest run uses only one of
the non-looping transitions of $s$ and possibly also the selfloop of $s$.
If $s$ is a must-state, a modest run can use an arbitrary combination of
transitions of $s$. Formally, $z$ is defined as
\begin{align*}
   z(s) = &\begin{cases}
   \{\{s\}\} \otimes \bigcup_{\substack{(s,\alpha,c) \in \delta\A, \\ s \notin 
   c}} \mathbf{z}(c) & \text{if } s \notin \must(S) \\[3ex]
   \{\{s\}\} \otimes \bigcup_{C \subseteq {\targets(s)}}  \mathbf{z}(\bigcup\limits_{c 
   \in C} c)      & \text{if } s \in \must(S)\text{,}
  \end{cases}
\end{align*}
where $\otimes$ is an auxiliary operation defined for each $W_1,W_2\subseteq
2^S$ as
\[
W_1 \otimes W_2 = \textstyle \bigcup_{\substack{c_1 \in W_1 \\
    c_2 \in W_2}} \{c_1 \cup c_2\}\text{,}
\]
$\targets(s)$ contains configurations reachable from $s$ in one step minus
the $s$ itself, i.e.
\[
\targets(s) = \left\{ c \in 2^{S\smallsetminus\{s\}} \mid (s, \alpha,c \cup
  \{s\}) \in \delta\A, \alpha \in \Sigma\right\}\text{,}
\]
and $\mathbf{z}$ is an auxiliary function defined for each configuration as
\[
\mathbf{z}(c) = \textstyle \bigotimes_{s \in c} z(s)
\qquad\textrm{and}\qquad
\mathbf{z}(\emptyset) = \{\emptyset\}.
\]

The function $y : S \rightarrow 2^{2^S}$ computes, for a given state $s$,
the sets of states that can potentially be reached from $s$ (not necessary
as one configuration) infinitely often by a modest run visiting $s$ at least
once.  Note that if $s$ is a must-state $z(s)=y(s)$.  Formally, $y$ is defined as
\begin{align*}
  y(s) = &\begin{cases}
   \bigcup_{\substack{(s,\alpha,c) \in \delta\A, \\ s \notin c}} \mathbf{y}(c) & 
   \text{if } s \notin \must(S) \\[3ex]
   \{\{s\}\} \otimes \bigcup_{C \subseteq {\targets(s)}}  \mathbf{z}(\bigcup\limits_{c 
   \in C} c)      & \text{if } s \in \must(S)\text{,}
  \end{cases}
\end{align*}
where $\mathbf{y}$ is an auxiliary function defined for each configuration
as
\[
\mathbf{y}(c) = \textstyle \bigotimes_{s \in c} y(s)
\qquad\textrm{and}\qquad
\mathbf{y}(\emptyset) = \{\emptyset\}.
\]

Finally, the set $\mZ$ of configurations $Z$ for which we construct the
generalized Rabin pairs $\mGR_Z$ of the corresponding TGDRA is defined as
$\mZ = \textstyle \bigcup_{c \in I} \mathbf{y}(c)$.
Note that the states that are not successors of any must-state are not
included in any element of the set $\mZ$.

%}}}
%{{{ Experimental results

\section{More Experimental Results}\label{sec:moreexp}
% zakomentovana tabulka parametrickych formuli, nyni zapracovana do
% talulky tab:ATVA
% \begin{table}[ht]
% \centering
% \subfigure[$\mathit{fair}(n) = \bigwedge_{i=1}^n(\G\F a_{i} \rightarrow \G\F b_{i})$]{
% \input{tabs/fair}}
% \subfigure[$R(n) = \bigwedge_{i=1}^n(\G\F a_i \lor \F\G a_{i+1})$]{\input{tabs/R}}
% \caption{Benchmark of parametric formulae}\label{tab:par}
% \end{table}

Table~\ref{tab:app} provides experimental results for four parametric
formulae, namely $\theta(n)=\neg ((\bigwedge_{i=1}^n\G\F a_{i})\rightarrow
\G (b_1 \rightarrow \F b_2))$ of~\cite{GO01}, its negation
$\neg\theta(n)$, $U(n)=(\ldots((a_1 \U a_2) \U a_3) \ldots ) \U a_n$ of~\cite{GH06}, and
$U_2(n)=a_1 \U (a_2 \U ( \ldots \U a_n)\ldots)$ of~\cite{GH06}. Note that
% $\theta(n)$ and $\neg\theta(n)$ are from LTL($\F$,$\G$), while $U(n)$ and
% $U_2(n)$ are outside this fragment.
$U(n)$ and $U_2(n)$ are outside the fragment LTL($\F$,$\G$).

In all the cases but one ($U(5)$), LTL3DRA generates DRA of the same size as
the other tools or smaller.  Finally, LTL3DRA runs slower than
\texttt{ltl2dstar} on $U_2(n)$, while it is the fastest on the other three
formulae.  We list all instances of the formulae for which at least one of
the tools produces a DRA before timeout.

%{{{ Table 3

\begin{table}[b]
\centering
\begin{tabular}{|c||r@{\quad}r@{\quad}|r@{\quad}r@{\quad}|r@{\quad}|}
\hline
\multicolumn{1}{ |c|| }{\multirow{2}{*}{Formula} } &
  \multicolumn{2}{ |c }{LTL3DRA} & 
  \multicolumn{2}{ |c| }{Rabinizer} &
  \multicolumn{1}{ |c| }{\texttt{~ltl2dstar~}}
  \\ \cline{2-6}
\multicolumn{1}{ |c||  }{} &
\multicolumn{2}{|r|}{\scriptsize DRA~~~TGDRA} & \multicolumn{2}{|r|}{\scriptsize DRA~~GDRA} & \multicolumn{1}{ |c|  }{\scriptsize DRA}  \\ \cline{1-6}
\hline%\hline
$\theta(1)$ & \textbf{3(1)}& 2 & \textbf{3(1)} & 10 & 4(1) \\
%\hline
$\theta(2)$ & \textbf{4(1)}& 2 & 5(1) & 20 & 8(1) \\
%\hline
$\theta(3)$ & \textbf{5(1)}& 2 & 7(1) & 40 & 12(1) \\
%\hline
$\theta(4)$ & \textbf{6(1)}& 2 & 9(1) & 80 & 16(1) \\
%\hline
$\theta(5)$ & \textbf{7(1)}& 2 & 11(1) & 160 & 20(1) \\
%\hline
$\theta(6)$ & \textbf{8(1)}& 2 & 13(1) & 320 & 24(1) \\
%\hline
$\theta(7)$ & \textbf{9(1)}& 2 & \multicolumn{2}{|c|}{timeout} & 28(1) \\
%\hline
$\theta(8)$ & \textbf{10(1)}& 2 & \multicolumn{2}{|c|}{timeout} & 32(1) \\
%\hline
$\theta(9)$ & \textbf{11(1)}& 2 & \multicolumn{2}{|c|}{timeout} & 36(1) \\
%\hline
$\theta(10)$ & \textbf{12(1)}& 2 & \multicolumn{2}{|c|}{timeout} & \multicolumn{1}{|c|}{timeout}\\
%\hline
$\theta(11)$ & \textbf{13(1)}& 2 & \multicolumn{2}{|c|}{timeout} & \multicolumn{1}{|c|}{timeout}\\
%\hline
$\theta(12)$ & \textbf{14(1)}& 2 & \multicolumn{2}{|c|}{timeout} & \multicolumn{1}{|c|}{timeout}\\
\hline
$\neg\theta(1)$ & \textbf{6(3)}& 2 & 8(3) & 10 & 7(2) \\
%\hline
$\neg\theta(2)$ & \textbf{12(4)}& 2 & 16(4) & 20 & 13(3) \\
%\hline
$\neg\theta(3)$ & \textbf{24(5)}& 2 & 32(5) & 40 & 25(4) \\
%\hline
$\neg\theta(4)$ & \textbf{48(6)}& 2 & 64(6) & 80 & 49(5) \\
%\hline
$\neg\theta(5)$ & \textbf{96(7)}& 2 & 128(7) & 160 & 97(6) \\
%\hline
$\neg\theta(6)$ & \textbf{192(8)}& 2 & \quad 256(8) & 320 & 193(7) \\
%\hline
$\neg\theta(7)$ & \textbf{384(9)}& 2 & \multicolumn{2}{|c|}{timeout}  & 385(8) \\
%\hline
$\neg\theta(8)$ & \textbf{768(10)}& 2 & \multicolumn{2}{|c|}{timeout}  & 769(9) \\
%\hline
$\neg\theta(9)$ & \textbf{1536(11)}& 2 & \multicolumn{2}{|c|}{timeout}  & 1537(10) \\
%\hline
$\neg\theta(10)$ & \textbf{3072(12)}& 2 & \multicolumn{2}{|c|}{timeout}  & 3073(11) \\
%\hline
$\neg\theta(11)$ & \quad\textbf{6144(13)}& 2 & \multicolumn{2}{|c|}{timeout}  & \multicolumn{1}{|c|}{timeout}\\
\hline
%$\mathit{Ul}_1$ & \textbf{3(1)}& 3 & \multicolumn{2}{|c|}{---} & \textbf{3(1)} \\
%\hline
$U(2)$ & \textbf{3(1)}& 3 & \multicolumn{2}{|c|}{---} & \textbf{3(1)} \\
%\hline
$U(3)$ & \textbf{5(1)}& 5 & \multicolumn{2}{|c|}{---} & \textbf{5(1)} \\
%\hline
$U(4)$ & \textbf{9(1)}& 9 & \multicolumn{2}{|c|}{---} & \textbf{9(1)} \\
%\hline
$U(5)$ & 24(1)& 24 & \multicolumn{2}{|c|}{---} & \textbf{17(1)} \\
%\hline
$U(6)$ & \textbf{68(1)}& 68 & \multicolumn{2}{|c|}{---} & \multicolumn{1}{|c|}{timeout}\\
%\hline
$U(7)$ & \textbf{212(1)}& 212 & \multicolumn{2}{|c|}{---} & \multicolumn{1}{|c|}{timeout}\\
%\hline
$U(8)$ & \textbf{719(1)}& 719 & \multicolumn{2}{|c|}{---} & \multicolumn{1}{|c|}{timeout}\\
\hline
%$U_2(1)$ & \textbf{3(1)}& 3 & \multicolumn{2}{|c|}{---} & \textbf{3(1)} \\
%\hline
$U_2(2)$ & \textbf{3(1)}& 3 & \multicolumn{2}{|c|}{---} & \textbf{3(1)} \\
%\hline
$U_2(3)$ & \textbf{4(1)}& 4 & \multicolumn{2}{|c|}{---} & \textbf{4(1)} \\
%\hline
$U_2(4)$ & \textbf{5(1)}& 5 & \multicolumn{2}{|c|}{---} & \textbf{5(1)} \\
%\hline
$U_2(5)$ & \textbf{6(1)}& 6 & \multicolumn{2}{|c|}{---} & \textbf{6(1)} \\
%\hline
$U_2(6)$ & \textbf{7(1)}& 7 & \multicolumn{2}{|c|}{---} & \textbf{7(1)} \\
%\hline
$U_2(7)$ & \textbf{8(1)}& 8 & \multicolumn{2}{|c|}{---} & \textbf{8(1)} \\
%\hline
$U_2(8)$ & \textbf{9(1)}& 9 & \multicolumn{2}{|c|}{---} & \textbf{9(1)} \\
%\hline
$U_2(9)$ & \textbf{10(1)}& 10 & \multicolumn{2}{|c|}{---} & \textbf{10(1)} \\
%\hline
$U_2(10)$ & \textbf{11(1)}& 11 & \multicolumn{2}{|c|}{---} & \textbf{11(1)} \\
%\hline
$U_2(11)$ & \textbf{12(1)}& 12 & \multicolumn{2}{|c|}{---} & \textbf{12(1)} \\
%\hline
$U_2(12)$ & \textbf{13(1)}& 13 & \multicolumn{2}{|c|}{---} & \textbf{13(1)} \\
%\hline
$U_2(13)$ & \multicolumn{2}{c|}{timeout} & \multicolumn{2}{|c|}{---} & \textbf{14(1)} \\
\hline
\end{tabular}
\smallskip
\caption{More parametric benchmarks.} 
\label{tab:app}
\end{table}

%}}}

%$\theta_n = \neg \bigwedge_{i=1}^n(\G\F a_{i})
%\rightarrow \G (b_1 \rightarrow \F b_2)$ 
%$\neg\theta_n = \bigwedge_{i=1}^n(\G\F a_{i})
%\rightarrow \G (b_1 \rightarrow \F b_2)$ 
%$\mathit{Ul}_{n} = (\ldots((a_1 \U a_2) \U a_3) \ldots ) \U a_n$
%
%$\mathit{Ur}_{n} = a_1 \U (a_2 \U ( \ldots \U a_n)\ldots)$ 

%}}}

\end{document}